\documentclass{lmcs} 
\pdfoutput=1
\usepackage[utf8]{inputenc}

\usepackage{lastpage}
\lmcsdoi{22}{1}{7}
\lmcsheading{}{\pageref{LastPage}}{}{}%
{Feb.~15,~2024}{Jan.~09,~2026}{}


\usepackage{hyperref}
\theoremstyle{plain} 

\usepackage{amsmath,amssymb,wasysym,amsthm}
\usepackage[all]{xy}
\usepackage{proof}
\usepackage{cmll}
\usepackage{stmaryrd}
\usepackage{xspace}
\usepackage{tikz}
\usepackage{tikz-cd}
\usepackage{bussproofs}
\usepackage[switch]{lineno}
\usepackage{mathrsfs}  
\usepackage{thmtools,thm-restate}
\usepackage{cmll}
\usepackage{bussproofs}
\usepackage{xspace}
\usepackage{proof}
\usepackage{stmaryrd}
\usepackage{tikz-cd}
\usepackage{pifont}

\EnableBpAbbreviations

\newcommand\Rel{\mathtt{Rel}}

\newcommand\R{\mathbb{R}}

\newcommand\N{\mathbb{N}}
\newcommand{\Sr}{\mathcal{S}}
\newcommand{\C}{\mathcal{C}}
\newcommand{\Lin}{\mathcal{L}}
\newcommand\D{\mathcal{D}}
\renewcommand\LL{\mathsf{LL}}
\newcommand\DiLL{\mathsf{DiLL}}
\newcommand\DiLLo{\mathsf{DiLL}_0}
\newcommand\MALL{\mathsf{MALL}}
\newcommand\DDiLL{\mathsf{D-DiLL}}

\newcommand\IndLL{\mathsf{IndLL}}
\newcommand\BLL{\mathsf{BLL}}
\newcommand\BSLL{\mathsf{B_{\mathcal{S}}LL}}
\newcommand\DBSLL{\mathsf{DB_{\mathcal{S}}LL}}
\newcommand\IDiLL{\mathsf{IDiLL}}
\newcommand{\der}{\mathsf{d}}
\newcommand{\contr}{\mathsf{c}}
\newcommand{\w}{\mathsf{w}}
\newcommand{\p}{\mathsf{p}}
\newcommand{\m}{\mathsf{m}}
\newcommand\bw{\bar{\mathsf{w}}}
\newcommand\bwi{\bar{\mathsf{w}}_I\xspace}
\newcommand\bc{\bar{\mathsf{c}}\xspace}
\newcommand\bd{\bar{\mathsf{d}}\xspace}
\newcommand\bdi{\bar{\mathsf{d}}_I\xspace}

\newcommand\wnD{\wn_D\xspace}
\newcommand\ocD{\oc_D\xspace}
\newcommand\wnu{\wn_{D_1} E}
\newcommand\wnd{\wn_{D_2} E}
\newcommand\wnud{\wn_{D_1 \circ D_2} E}
\newcommand\ocu{\oc_{D_1} E}
\newcommand\ocd{\oc_{D_2} E}
\newcommand\ocud{\oc_{D_1 \circ D_2} E}

\newcommand\br[1]{\llbracket #1 \rrbracket}

\newenvironment{bpt}
  {\leavevmode\hbox\bgroup}
  {\DisplayProof\egroup}

  \newenvironment{smallreduc}[2]
  {\begin{center}#1\quad$\lto$\quad#2}
  {\end{center}}
  
  \newenvironment{bigreduc}[2]{
  \begin{minipage}[t]{1.0\linewidth}
  \begin{flalign*}
  #1
  \quad\lto\quad
  &&
  \end{flalign*}
  \begin{flalign*}
  &&
  #2
  \end{flalign*}
  \end{minipage}
  }

\newcommand\lto{\leadsto}

\newcommand\Rc{\mathcal{R}}

\newcommand\Ci{\mathcal{C}^\infty}

\newcommand\Xo{X^{\left(\omega\right)}}

\newcommand\nb[1]{{#1}^{\bot}}

\newcommand{\id}{\mathsf{id}}

\newcommand{\sem}[1]{\llbracket #1 \rrbracket}

\newcommand*\fonction[5]{#1 \colon \left\{\begin{alignedat}{2}  #2  &\to      #3\\
    #4    &\mapsto  #5
    \end{alignedat} \right. \kern-\nulldelimiterspace}

\allowdisplaybreaks

\begin{document}

\title{Unifying Graded Linear Logic and Differential Operators}

\author[F.~Breuvart]{Flavien Breuvart\lmcsorcid{0000-0002-0425-3270}}[a]
\author[M.~Kerjean]{Marie Kerjean\lmcsorcid{0000-0001-6141-6251}}[b,a]
\author[S.~Mirwasser]{Simon Mirwasser\lmcsorcid{0000-0002-2714-8123}}[a]
\thanks{The authors are grateful to the reviewers for providing with numerous feedback and suggestions, which greatly improved the paper.}

\address{Université Sorbonne Paris Nord, 99 Avenue Jean Baptiste Clément, 93430 Villetaneuse, France}	
\email{breuvart@lipn.fr, kerjean@lipn.fr, mirwasser@lipn.fr}  

\address{CNRS}	





\begin{abstract}
  Linear Logic refines Classical Logic by taking into account the resources used during the proof and program computation. In the past decades, it has been extended to various frameworks. The most famous are indexed linear logics which can describe the resource management or the complexity analysis of a program. From another perspective, Differential Linear Logic is an extension which allows the linearization of proofs. In this article, we merge these two directions by first defining a differential version of Graded linear logic: this is made by indexing exponential connectives with a monoid of differential operators. We prove that it is equivalent to a graded version of previously defined extension of finitary differential linear logic. We give a denotational model of our logic, based on distribution theory and linear partial differential operators with constant coefficients.
\end{abstract}

\maketitle

\section{Introduction}\label{sec:intro}

Linear logic ($\LL$)~\cite{girardll} provides a framework for studying the use of resources in proofs and programs.
It was developed by enriching the syntax of proofs with new constructions observed in denotational models of $\lambda$-calculus~\cite{girard_normal_1988}. 
In this work, we present a first combination of two a priori distinct versions of linear logic: graded logics~\cite{bll} and differential linear logic~\cite{dill}. Adopting a semantically oriented approach,  we build on  Kerjean's former works \cite{ddill} and focus on a particular model of this logic involving differential operators. This work therefore provides both new models for graded logics and a fresh perspective on differentiation in $\LL$.

\subsection{Historical context}
The fundamental linear decomposition of $\LL$ is the decomposition of the usual non-linear implication $\Rightarrow$ into a linear one $\multimap$ from a set of resources represented by the new connective~$\ oc$:
$ (A \Rightarrow B) \equiv (\oc A \multimap B). $ 
The exponential connective $\oc$ introduces non-linearity in the context of linear proofs and encapsulates the notion of \emph{resource} usage.
 This notion has been refined into \emph{parametrised exponentials}~\cite{bll,ill,effcoeff-grading,GhicaSmith}, where exponential connectives are indexed by annotations that specify different behaviors.
Bounded Linear Logic ($\BLL$)~\cite{bll} was introduced as the first attempt to use typing systems for the purpose of complexity analysis. In essence, indices are used to track the resource usage within proofs. For instance, a proof of $\oc_n A \to B$ will use $n$ times the hypothesis $A$.
More generally, $\BLL$ extends $\LL$ with several exponential connectives which are indexed by \emph{polynomially bounded intervals}. 
Since then, some other indexations of $\LL$ have been developed for many purposes, for example $\IndLL$~\cite{ill}, where the exponential modalities are indexed by some functions, or the graded logic 
$\BSLL$~\cite{corequantitative,GhicaSmith,melliesmonads} where they are indexed by the elements of a semiring $\mathcal{S}$. 
This theoretical development has found applications in the field of programming languages~\cite{haskell,diffprivacy}.

Differential linear logic~\cite{dill} ($\DiLL$) was also invented following a denotational study of a model of $\LL$~\cite{ehrhard_finiteness_2005}. However, it employs of an a priori distinct approach to linearity than graded logics, and is based on the denotational semantics of linear proofs in terms of linear functions.
In the syntax of $\LL$, the dereliction rule stipulates that if a proof is linear, it can then be considered non-linear. 

In order to capture differentiation, $\DiLL$ is based on a codereliction rule, which is the syntactical opposite of the dereliction.
It states that from a non-linear proof (or a non-linear function) one can extract a linear approximation of it, which, in terms of functions, is exactly the differential. Notice that here, the analogy with resources does not work as well. Our focus has shifted from tracking resources  forcing to some proofs to make a linear use of some of their resources.
Models of $\DiLL$ naturally interpret the codereliction by different kinds of differentiation~\cite{ehrhard_kosequence_2002,blute_convenient_2012}.

A first step towards merging the graded and the differential extension of $\LL$ was made by Kerjean in 2018~\cite{ddill}. In this paper, she shifts from the use of differentiation in $\DiLL$ to the use of linear partial differential operators with constant coefficients (LPDOcc). These are pervasive in applied mathematics and characterise many physical phenomena, such as the heat equation or the wave equation.  This is a significant step forward in our aim to make the theory of programming languages and functional analysis closer, with a Curry-Howard perspective. This work aligns logic and applied mathematics more closely and might result in formal systems where types can determine the correctness of an approximation to a differential equation.
To be more precise, the paper defines an extension of $\DiLL$, named $\DDiLL$, in which the exponential connectives $\wn$ and $\oc$ are indexed with a \emph{fixed} LPDOcc $D$. 
In this context, we interpret the formulas $\ocD A$ and $\wnD A$ as, respectively, spaces of functions or distributions which are solutions of the differential equation induced by $D$. 
The dereliction and codereliction rules then represent, respectively, the resolution of a differential equation and the application of a differential operator.
In this work, we will generalise $\DDiLL$ to a logic indexed by a monoid of LPDOcc.

\subsection{Contributions} 
This work considerably generalises and consolidates the extension of $\DiLL$ to differential operators sketched in~\cite{ddill}. It extends $\DDiLL$, enabling proofs to deal with all LPDOcc and combine their actions. It corrects $\DDiLL$ as the denotational interpretation of indexed exponential $\wn_D$ and $\oc_D$ are changed, leaving the interpretation of \emph{inference rules} unchanged but reversing their type in a way that is now compatible with graded logics. Finally, this work consolidates $\DDiLL$ by proving a cut-elimination procedure in the graded case. This is achieved by making use of an algebraic property on the monoid of LPDOcc. 

Let us give a summary of our results. The central rule in $\DiLL$ is the codereliction $\bd$, which symmetric to the dereliction $\der$ already present as an exponential rule in $\LL$. 
 \vspace{\baselineskip}
\begin{center}
  \AXC{$\Gamma, A \vdash \Delta$}
  \RL{$\der$}
  \UIC{$\Gamma, \oc A \vdash \Delta$}
  \DisplayProof
  \qquad
    \qquad
  \AXC{$\Gamma\vdash\Delta,A$}
  \RL{$\bd$}
  \UIC{$\Gamma\vdash\Delta,\oc A$}
  \DisplayProof
\end{center}
 \vspace{\baselineskip}
The intuition behind $\der$ is that it maps a linear function $\ell : A \multimap B$ to the same function but with the type of a non-linear map $\ell :  \oc A \multimap B$. The intuition behind $\bd$ is that it maps a vector $ v : A$ to the operation $ f \mapsto D_0(f)(v)$, mapping a function to its differential at $0$ along the vector $v$. Then, $\bd(v)$ is a distribution, that is, a linear map acting on smooth functions, and typed by $\oc A$. The central observation behind the construction of $\DDiLL$ was that $\bd$ could be assimilated to $ \delta_v \mapsto \delta_v  \circ D_0(\_)(\_)$, the operation mapping a distribution to the same distribution composed with the differential operator differentiating at $0$. Dereliction could in turn be understood as a function mapping $\ell$ to the solution of the equation $D_0(\_) = \ell$, itself in this particular case. In the present work, we take inspiration from graded linear logic to correctly type these rules. While indices are traditionally used to track the usage of resources, here they will track the usage of differential operators, where $D$ is another differential operator. 
 \vspace{\baselineskip}
 \begin{center}
  \AXC{$\Gamma, \oc A \vdash \Delta$}
  \RL{$\der_D$}
  \UIC{$\Gamma, \oc_D A \vdash \Delta$}
  \DisplayProof
  \qquad
    \qquad
  \AXC{$\Gamma\vdash\Delta, \oc A$}
  \RL{$\bd_D$}
  \UIC{$\Gamma\vdash\Delta, \oc_D A$}
  \DisplayProof
\end{center}
 \vspace{\baselineskip}
The indexed codereliction $\bd_D$ composes a distribution with a differential operator while the indexed dereliction $\der_D$ solves a differential equation. It can be interpreted in a denotational model when the equation $D (\_) = g$ always has a solution. This is the case in particular for linear partial differential operators with constant coefficients. They also enjoy another property, as they behave particularly well with respect to the convolution of distributions. If $D_1$ and $D_2$ are linear partial differential operators with constant coefficients, and $\phi$ and $\psi$ are two distributions, then $(D_1 \ast \phi) \circ (D_2 \ast \psi) = (D_1 \circ D_2) \circ (\phi \ast \psi)$. This means that the set of linear partial differential operators with constant coefficients is well-behaved with respect to the co-contraction $\bc$ rule of $\DiLL$, interpreted by convolution.
 \vspace{\baselineskip}
\begin{center}
  \AXC{$\Gamma, \oc A , \oc A \vdash \Delta $}
  \RL{$\contr$}
  \UIC{$ \Gamma, \oc A \vdash\Delta$}
  \DisplayProof
\qquad
  \AXC{$\Gamma_1\vdash\Delta_1, \oc A$}
  \AXC{$\Gamma_2\vdash \Delta_2,  \oc A$}
  \RL{$\bc$}
  \BIC{$\Gamma_1,\Gamma_2\vdash\,\Delta_1, \Delta_2,  \oc A$}
  \DisplayProof
\end{center}
 \vspace{\baselineskip}
In graded linear logic, the contraction rule typically introduces the additive law on the semi-ring of indices. If graded contraction is to be typed as graded co-contraction, this means that in our case, we will consider the set of linear partial differential operators with constant coefficients endowed with composition as an additive law.
 \vspace{\baselineskip}
\begin{center}
  \AXC{$\Gamma, \oc_{D_1} A , \oc_{D_2} A \vdash \Delta $}
  \RL{$\contr$}
  \UIC{$ \Gamma, \oc_{D_1 \circ D_2} A \vdash\Delta$}
  \DisplayProof
\qquad
\AXC{$\Gamma_1\vdash\Delta_1, \oc_{D_1} A$}
\AXC{$\Gamma_2\vdash \Delta_2,  \oc_{D_2} A$}
\RL{$\bc$}
\BIC{$\Gamma_1,\Gamma_2\vdash\,\Delta_1, \Delta_2,  \oc_{D_1\circ D_2} A$}
\DisplayProof
\end{center}
 \vspace{\baselineskip}
 One other crucial property is the fact that each linear partial differential operator with constant coefficients $D$ has a \emph{fundamental solution} $\Phi_D$. This is a distribution such that $D(\Phi_D) = \delta_0$, where $\delta_0$ denotes the the Dirac distribution at 0. This property implies that for each operator $D_1$ and $D_2$, if $f$ is in $\wn_{D_1} A$ then $\Phi_{D_2}\ast f$ is in $\wn_{D_1 \circ D_2}A$. And if $\psi$ is in $\oc_{D_1}A$, then $\psi \circ D_2$ is in $\oc_{D_1 \circ D_2}A$. These results give us a way to interpret the following rules.
 \vspace{\baselineskip}
 \begin{center}
  \AXC{$\Gamma, \oc_{D_1} A \vdash \Delta$}
  \RL{$\der_I$}
  \UIC{$\Gamma, \oc_{D_1\circ D_2} A \vdash \Delta$}
  \DisplayProof
  \qquad
    \qquad
  \AXC{$\Gamma\vdash\Delta, \oc_{D_1} A$}
  \RL{$\bd_I$}
  \UIC{$\Gamma\vdash\Delta, \oc_{D_1\circ D_2} A$}
  \DisplayProof
\end{center}
 \vspace{\baselineskip}
 These rules are generalizations of $\der_D$ and $\bd_D$, with $D_1 = id$.
\par In our work, we begin by constructing a differential version of $\BSLL$ without promotion, by giving symmetric costructural rules to every structural rule. We prove a cut-elimination theorem, which depends on an algebraic property of the semiring used to index the rules, and define a procedure which mimics partly that of $\DiLL$ or $\BSLL$. However, some parts of the cut elimination procedure are completely different, since we need to deal with new interactions between inference rules. These are the cases involving $\der_I$ and $\bd_I$, which have not been studied in a differential setting. We make explicit a relational model for this syntax in Section~\ref{ssec:semantics}. These developments result in a syntax that we call $\DBSLL$, presented below:
 \vspace{\baselineskip}
  \begin{center}
 \vbox{
          \begin{center}
              \AXC{$\vdash \Gamma$}
              \RL{$\w$}
              \UIC{$\vdash \Gamma,\wn_0 A$}
              \DisplayProof
              \qquad
              \AXC{$\vdash \Gamma,\wn_x A,\wn_y A$}
              \RL{$\contr$}
              \UIC{$\vdash \Gamma,\wn_{x + y} A$}
              \DisplayProof
              \qquad
              \AXC{$\vdash \Gamma,\wn_x A$}
              \AXC{$x \leq y$}
              \RL{$\der_I$}
              \BIC{$\vdash \Gamma,\wn_y A$}
              \DisplayProof
              \qquad
              \AXC{$\vdash \Gamma, A$}
              \RL{$\der$}
              \UIC{$\vdash \Gamma, \wn A$}
              \DisplayProof
          \end{center}
          \begin{center}
              \AXC{}
              \RL{$\bw$}
              \UIC{$\vdash \oc_0 A$}
              \DisplayProof
              \qquad
              \AXC{$\vdash \Gamma, \oc_x A$}
              \AXC{$\vdash \Delta, \oc_y A$}
              \RL{$\bc$}
              \BIC{$\vdash \Gamma,\Delta, \oc_{x + y} A$}
              \DisplayProof
              \qquad
              \AXC{$\vdash \Gamma,\oc_x A$}
              \AXC{$x \leq y$}
              \RL{$\bdi$}
              \BIC{$\vdash \Gamma,\oc_y A$}
              \DisplayProof
              \qquad
              \AXC{$\vdash \Gamma, A$}
              \RL{$\bd$}
              \UIC{$\vdash \Gamma, \oc A$}
              \DisplayProof
          \end{center}
      }
  \end{center}
 \vspace{\baselineskip}
Then we provide a graded version of $\DiLL$ without promotion, that we call $\IDiLL$:
 \vspace{\baselineskip}
    \begin{center}
        \AXC{$\vdash \Gamma$}
        \RL{$\w_I$}
        \UIC{$\vdash \Gamma, \wnD A$}
        \DisplayProof
        \qquad\qquad
        \AXC{$\vdash \Gamma,\wn_{D_1} A,\wn_{D_2} A$}
        \RL{$\contr$}
        \UIC{$\vdash \Gamma,\wn_{D_1 \circ D_2} A$}
        \DisplayProof
        \qquad\qquad
        \AXC{$\vdash \Gamma,\wn_{D_1} A$}
        \RL{$\der_I$}
        \UIC{$\vdash \Gamma,\wn_{D_1 \circ D_2} A$}
        \DisplayProof
    \end{center}
    \begin{center}
        \AXC{}
        \RL{$\bwi$}
        \UIC{$\vdash \ocD A$}
        \DisplayProof
        \qquad\qquad
        \AXC{$\vdash \Gamma, \oc_{D_1} A$}
        \AXC{$\vdash \Delta, \oc_{D_2} A$}
        \RL{$\bc$}
        \BIC{$\vdash \Gamma, \Delta, \oc_{D_1 \circ D_2} A$}
        \DisplayProof
        \qquad\qquad
        \AXC{$\vdash \Gamma,\oc_{D_1} A$}
        \RL{$\bdi$}
        \UIC{$\vdash \Gamma,\oc_{D_1 \circ D_2} A$}
        \DisplayProof
    \end{center}
 \vspace{\baselineskip}
$\IDiLL$ generalizes and corrects $\DDiLL$. We show that it consists of admissible rules of $\DBSLL$ and that rules of $\DBSLL$, except $\der$ and $\bd$, are admissible in $\IDiLL$. We connect back to our first intuitions by providing a denotational model of $\IDiLL$ based on distributions and smooth functions, where the set of indices is the monoid of linear partial differential operators with constant coefficients endowed with composition as an additive law.  Finally, we discuss the addition of an indexed promotion to differential $\BSLL$ and possible definitions for a semiring of differential operators.

\subsection{Outline}
We begin this paper in Section~\ref{sec:LL} by reviewing Differential Linear Logic and its semantics in terms of functions and distributions. We also recall the definition of $\BSLL$. Section~\ref{sec:dbsll} focuses on the definition of an extension of $\BSLL$, in which we construct a finitary differential version for it and prove a cut-elimination theorem and exhibit a relational model for this syntax. Then, Section~\ref{sec:idill} generalises $\DDiLL$ into a framework with several indexes and shows that it corresponds to our finitary differential $\BSLL$ indexed by a monoid of LPDOcc. It formally constructs a denotational model for it based on spaces of functions and distributions. This gives, in particular, a new semantics for $\BSLL$. Finally, the conclusion discusses related works as well as the addition of an indexed promotion to differential $\BSLL$, and potential definitions for a semiring of differential operators.
\newpage
\tableofcontents
 
  \section{Linear logic and its extensions}\label{sec:LL}
  Linear Logic refines Intuitionistic Logic by introducing a notion of linear proofs. Formulas are defined according to the following grammar (omitting neutral elements which do not play a role here):
  \[ A,B := A \otimes B \mid A \parr B \mid A \with B \mid A \oplus B \mid \wn A \mid \oc A \mid \cdots.\]
  
  The linear negation $\nb{(\_)}$ of a formula is \emph{defined} on the syntax and is involutive, with in particular $\nb{(\oc A)} := \wn \nb{(A)}$.
  The connector $\oc$ enjoys structural rules, respectively called weakening $\w$, contraction $\contr$, dereliction $\der$ and promotion $\p$:
  \begin{center}
  \AXC{$\Gamma \vdash \Delta$}
  \RL{$\w$}
  \UIC{$\Gamma, \oc A \vdash \Delta $}
  \DisplayProof
  \qquad
  \AXC{$\Gamma, \oc A , \oc A \vdash \Delta $}
  \RL{$\contr$}
  \UIC{$ \Gamma, \oc A \vdash\Delta$}
  \DisplayProof
  \qquad
  \AXC{$\Gamma, A \vdash \Delta$}
  \RL{$\der$}
  \UIC{$\Gamma, \oc A \vdash \Delta$}
  \DisplayProof
  \qquad
  \AXC{$\oc \Gamma  \vdash \wn\Delta, A$}
  \RL{$\p$}
  \UIC{$\oc \Gamma  \vdash \wn \Delta, \oc A $}
  \DisplayProof 
  \end{center}
  These structural rules can be understood in terms of resources: a proof of $A \vdash B$ uses exactly once the hypothesis $A$ while a proof of $\oc A \vdash B$ might use $A$ an arbitrary number of times. Notice that the dereliction allows one to forget the linearity of a proof by making it non-linear. Weakening means that the use of $\oc A$ can mean the use of no resources of type $A$ at all, while the contraction rule represents the glueing of resources: using twice an arbitrary amount of data of type $A$ corresponds to using once an arbitrary amount of data of type $A$.  
  \begin{rem}
      The exponential rules for $\LL$ are recalled here in a two-sided flavour, making their denotational interpretation in Section \ref{subsec:distrib} easier. However, we always consider a \emph{classical} sequent calculus, and the new $\DBSLL$ will be introduced later in a one-sided flavour to lighten the formalism.
  \end{rem}
  
These resource intuitions are challenged by Differential Linear Logic.  Differentiation is introduced through a new ``codereliction'' rule $\bd$, which is symmetrical to $\der$ and allows to linearise a non-linear proof~\cite{dill}. To express the cut-elimination with the promotion rule, other costructural rules are needed, which find a natural interpretation in terms of differential calculus. 

Note that the first version of $\DiLL$, called $\DiLL_0$, does not feature the promotion rule, which was introduced in later versions~\cite{pagani_cut-elimination_2009}. The exponential rules of $\DiLL_0$ are then $\w,  \contr, \der$ with the following coweakening $\bw$, cocontraction $\bc$ and codereliction~$\bd$ rules, given here in a one-sided flavour.

  \begin{center}
  \AXC{}
  \RL{$\bw$}
  \UIC{$\vdash \oc A$}
  \DisplayProof
  \qquad
  \AXC{$\vdash\Gamma, \oc A$}
  \AXC{$\vdash \Delta,  \oc A$}
  \RL{$\bc$}
  \BIC{$\vdash\Gamma, \Delta,  \oc A$}
  \DisplayProof
  \qquad
  \AXC{$\vdash\Gamma,A$}
  \RL{$\bd$}
  \UIC{$\vdash\Gamma, \oc A$}
  \DisplayProof
  \end{center}
  
  In the rest of the paper, as a support for the semantical interpretation of $\DiLL$, we denote by~$D_a f$ the differential\footnote{Differentiation at higher-order might be tricky, and to be precise, we are here using the Gateaux differential of a function. Most of the time however, we will only consider smooth functions $f : \R^n \to R$ where these distinctions do not matter.} of a function $f$ at a point $a$, that is:
  \begin{center}
  $\displaystyle D_a f : v \mapsto \lim_{h \to 0} \frac{f(a +hv) - f(a)}{h}$
  \end{center}
  \subsection{Distribution theory as a semantical interpretation of DiLL}\label{subsec:distrib}
  $\DiLL$ originates from vectorial refinements of models of $\LL$~\cite{ehrhard_finiteness_2005}, which mainly keep their discrete structure. 

Consider the interpretation $f : A \Rightarrow B$ of a proof of $ \oc A \vdash B$. Then by cut-elimination, the codereliction creates a proof $\bd; f : A \multimap B$. Other exponential rules also have an easy functional interpretation by pre-composition:
\begin{itemize}
\item $\bw;f : 1 \multimap B$ maps $1$ to $f(0)$,
\item $\bc;f : A \times A \Rightarrow B$ maps $(x,y)$ to $f(x +y)$,
\item for a function $g :A \times A \Rightarrow B$, $\contr;g$ maps $x :A$ to $g(x,x)$
\item for a pointed object $b : 1 \Rightarrow B$, $\w;b$ maps any $x :A$ to $b:B$,
\item dereliction maps a linear function $\ell: A \multimap B$ to the same function with a non-linear type : $\ell : A \Rightarrow B$.
\end{itemize} 
 
\paragraph{Smooth models of Classical $\DiLL$}
The above interpretations all have an intuitionistic flavour: they are valid up to composition with a non-linear function, and correspond to bilateral rules for $\w$, $\contr$ and $\der$ as presented above. In a model interpreting the involutive linear duality of Classical Linear Logic, exponential rules have a stand-alone interpretation, and distribution theory provides particularly relevant intuitions.
Exponential connectives and rules of $\DiLL$ can be understood as operations on smooth functions or distributions~\cite{schwartz_theorie_1966}.
When smooth functions rightfully interpret proofs of non-linear sequents $\oc A \vdash B$, distribution spaces give an interpretation for the exponential formula $\oc A$. 

 In the whole paper, $(\_)':= \Lin(\_,\R)$ is the dual of a (topological) vector space, and \emph{distributions with compact support} are by definition linear continuous maps on the space of smooth scalar maps, that is elements of $\left(\Ci(\R^n,\R)\right)'$. Distributions are sometimes described as  ``generalised functions''. Indeed, any smooth function with compact support $g \in \Ci_c(\R^n,\R)$ acts as a distribution $T_g \in \left(\Ci(\R^n,\R)\right)'$ with compact support, through integration:
  $T_g : f \mapsto \int g f $. It is indeed a distribution, as it acts linearly (and continuously) on smooth functions.
   Let us recall the notation for the Dirac operator, which is a distribution with compact support and is used a lot in the rest of the paper: 
  $\delta: v \in \R^n  \mapsto (f \mapsto f(v)) \in \left(\Ci(\R^n,\R)\right)'.$
  
  Recently, Kerjean~\cite{ddill} gave an interpretation of the connective $\wn$ by a space of smooth scalar functions, while $\oc$ is interpreted as the space of linear maps acting on those functions, that is a space of \emph{distributions with compact support\footnote{Here we are using smooth function without any condition on their support, and hence the distributions are said to have compact support. This is due to the necessity of interpreting the dereliction rule $\der : E' \to \wn (E')$, embedding linear functions (which do not have a compact support) into smooth functions.}}:
   \vspace{\baselineskip}
  \begin{center}
  $ \sem{? A} := \Ci(\sem{A}', \R) \qquad\qquad \sem{\oc A} := \Ci(\sem{A},\R)' $. 
  \end{center}
   \vspace{\baselineskip} 
While the language of distributions applies to all models of $\DiLL$, as noticed by Ehrhard on K\"othe spaces~\cite{ehrhard_kosequence_2002}, the focus of this model was to find smooth infinite-dimensional models of $\DiLL$, making the interpretations of maps and formulas objects of distributions theory as studied in the literature.  Another focus on top of that was to construct a model of \emph{classical} $\DiLL$, in which objects are invariant under double negation. 
 We will not dive into the details of these definitions, see~\cite{Jarchow} for a fine-grained exposition, but the reader should keep in mind that the formulas are always interpreted as \emph{reflexive topological vector spaces}. The model of functions and distribution is thus a model of \emph{classical} $\DiLL$, in which $\sem{(\_)^{\bot}}:=(\_)'$.

A locally convex and separated vector space is said to be \emph{reflexive} when it is linearly homeomorphic to its double dual: \[E \simeq E''.\]
This means two things. On the one hand, $E$ and $E''$ are the same vector spaces, meaning any linear form $\phi \in \Lin (E',\R)$ corresponds in fact to a point $x \in E$:
\[ \phi = (\ell \in E' \mapsto \ell(x)) .\]
On the other hand, $E$ and $E''$ must correspond topologically. This is an intricate issue. Traditionally, $E'$ is endowed with the topology of uniform convergence on bounded subsets of $E$, and likewise $E''$ is endowed with the topology of uniform convergence on bounded subsets of $E'$. The fact that this topology corresponds to the original one on $E$ is called a barrelledness condition, saying that absorbing sets of $E$ are in fact neighbourhoods of $0$. This idea is hard to grasp as it holds trivially on any finite-dimensional space and on any Hilbert space. One should just know that by default, Banach spaces are not reflexive. Moreover, the subclass of reflexive topological vector spaces does not enjoy good stability properties: they are not stable by tensor product, making them unqualified to be a model of Linear Logic.

\paragraph{Metrizable spaces and their dual}

In fact, most spaces of functions are not normed but \emph{metrizable}: as a consequence, we will not make use of Banach spaces but of complete metrizable spaces, also called \emph{Fréchet} spaces. Duals of Fréchet spaces are not metrizable: they are called DF-spaces, and spaces of distributions are particular examples of DF-spaces. An interesting condition to add is the \emph{nuclearity} of the space~\cite{grothendieck_produits_1966}, corresponding to a condition on topological tensor products. Indeed, any nuclear Fréchet space or nuclear DF-space is a reflexive space! Spaces of functions as $\Ci(\R^n,\R)$ and spaces of distributions as $\Ci(\R^n,\R)'$ are examples of nuclear Fréchet spaces and nuclear DF-spaces, respectively. 
We refer the interested reader to the literature~\cite{ddill,Jarchow}.
 
We thus interpret formulas of $\DiLL_0$ by spaces which are either complete Nuclear DF-spaces or Nuclear Fréchet spaces. Every construction of $\MALL$ applies to both classes of formulas: $\otimes$ and $\parr$ are both interpreted by the completion of the projective topological tensor product, and $\with$ and $\oplus$ are both interpreted by the biproduct on topological vector spaces. The connectives $\otimes$ and $\parr$ are both interpreted by the same connective as a result of nuclearity. 

\paragraph{Vector spaces of finite dimension} Our model is only a model of $\DiLLo$ where promotion is not interpreted. Even more, the exponential connective $\oc$ and $\wn$ apply only to finite-dimensional vector spaces. The reason is that there is no easy way to make $\Ci(E,\R)$ a Fréchet space when $E$ is a DF-space of infinite dimension. Hence, while we do have an interpretation of $\wn \R^n$ and $\oc \R^n$, we do not have a general interpretation for $\wn F$ and $\oc E$ for $E$ and $F$ DF and Fréchet nuclear spaces, respectively. Another reason for that is that we will be interested in applying partial differential operators to these spaces. These are only associated with a canonical base for the space, and traditionally defined on finite-dimensional vector spaces. Had we not had this restriction, we would have had to interpret to restrict our calculus to a version of polarised Linear Logic~\cite{laurent_phd}. In this case, we would interpret positive formulas (left stable by $\otimes$ $\oc$) as nuclear DF-spaces, while Negative formulas (left stable by $\parr$ $\wn$) are interpreted by nuclear Fréchet spaces. A higher-order version of our calculus is work in progress, based on higher-order functions defined on Fréchet and DF-spaces~\cite{Gannoun_dualite}.

\paragraph{Interpreting $\DiLLo$}
 We now describe the interpretation of every exponential rule of $\DiLL$ in terms of functions and distributions, through the following natural transformations. In the whole paper, $E$ and $F$ denote locally convex topological vector spaces, which will represent the interpretation $\sem{A}$ and $\sem{B}$ of formulas $A,B$ of $\DiLL$. Beware that we do not interpret higher-order functions, so that $\wn E$ and $\oc E$ will always be interpreted by $\Ci (\R^n, \R)$ and $\Ci (\R^n, \R)'$ for some $n$. The $\MALL$ connectives are easily interpreted on any locally convex topological vector space: $\otimes$ and $\parr$ are both interpreted by the completed projective tensor product (they are isomorphic thanks to the nuclearity condition), and $\oplus$ and $\with$ are both interpreted by the binary product (which is isomorphic to the biproduct in topological vector spaces).
 
For the sake of readability, we will denote the natural transformations (\textit{e.g.} $\der,\bd$) by the same label as the deriving rule they interpret, and likewise for connectors (\textit{e.g.} $\wn, \otimes, \oc$) and their associated functors. We use the fact that, as we are working with reflexive spaces, one can, without loss of generality, define a linear map $\ell : E \to F$ by its dual $\ell' : F' \to E'$ : taking the dual of $\ell'$ will give us $\ell$. 
  \begin{itemize}
  \item The weakening $\w : \R \to \wn E$ maps $1 \in \R$ to the constant function at $1$, while the coweakening $\bw : \R \to \oc E$ maps $1 \in \R$ to Dirac distribution at $0$: $\delta_0 : f \mapsto f(0)$.
  \item The dereliction $\der : E' \to \wn (E')$ maps a linear function to itself (but by typing it as a non-linear map, it forgets its linearity). The codereliction $\bd : E \to \oc E$ maps a vector $v$ to the distribution mapping a function to its differential at $0$ according to the vector $v$ : \begin{center}
  $\der : \ell \mapsto \ell \qquad \bd : v \mapsto \left( D_0 (\_)(v) : f \mapsto D_0 (f)(v) \right).$
  \end{center}
  \item The contraction $\contr : \wn E \otimes \wn E \to \wn E$ maps two scalar functions $f,g$ to their pointwise multiplication  $f \cdot g : x \mapsto f(x) \cdot g(x) $.
  \item  The cocontraction $\bc : \oc E  \otimes \oc E \to \oc E$ maps two distributions $\psi$ and $\phi$ to their \emph{convolution product} $\psi \ast \phi : f \mapsto \psi \left( x \mapsto \phi (y \mapsto f(x+y)) \right)$, which is a commutative operation over distributions.
 \item The promotion rule also has an easy interpretation in terms of distributions. This rule is interpreted thanks to the digging operator $\mu : \oc E \to \oc \oc E; \delta_x \mapsto \delta_{\delta_x}$. However, notice that the model we are describing is only a model of $\DiLLo$, as the object $\oc \oc E$ is not trivial to construct. As said before, only $\oc \R^n$ has an easy interpretation. 
  \end{itemize}
  These interpretations are natural while trying to give a semantics of a model with smooth functions and distributions. An important point is that sequent interpretation will often make a transparent scalar product appear. Indeed, the interpretation of $\otimes$ and $\parr$ on $\R = \sem{\bot}= \sem{1}$ is nothing but a plain product, denoted by a dot: $(a\otimes b \mapsto (a \cdot b))$.

  The fact that the contraction is interpreted by the scalar product is a direct consequence of Schwartz's kernel theorem. Indeed, in a differential category, the interpretations for $\contr$ and $\bd$ are direct consequences of the existence of a biproduct $\oplus \simeq \times$ and of the strong monoidality of $\oc $ : $\oc (A \times B) \simeq \oc A  \otimes B$ \cite{fioredill}. In a model where $\oc A$ is interpreted as a space of distributions, this last axiom is exactly Schwartz's kernel theorem, whose proof consists of showing the surjectivity of the map :
  $ (f \otimes g ) \in \wn A \parr \wn B \mapsto (f \cdot g) \in \wn ( A \times B)$. This is explained in more detail in~\cite{ddill}.
Weakening and co-weakening are respectively the neutral elements for the contraction and cocontraction laws.
  
   \subsection{Differential operators as an extension of $\DiLLo$}\label{ssec:ddill}

  A first advance in merging the graded and the differential extensions of $\LL$ was made by Kerjean in 2018~\cite{ddill}. 
  In this paper, she defines an extension of $\DiLL$ named $\DDiLL$. This logic is based on a \emph{fixed single} linear partial differential operator $D$, which appears as a single index in exponential connectives $\ocD$ and $\wnD$.
   
   The abstract interpretation of $\wn$ and $\oc$ as spaces of functions and distributions, respectively, allows us to generalise them to spaces of solutions and parameters of differential equations. To do so, we generalise the action of $D_0(\_)$ in the interpretation of $\bd$ to another differential operator $D$. 
   The interpretation of $\bd$ then corresponds to the application of a differential operator while the interpretation of $\der$ corresponds to the resolution of a differential equation (which is $\ell$ itself when the equation is $D_0(\_)=\ell$, but this is specifically due to the idempotency of $D_0$).
  
   In $\DDiLL$, the exponential connectives can be indexed by a fixed differential operator. It admits a denotational semantics for a specific class of those whose resolution is straightforward, thanks to the existence of a fundamental solution.
   A \emph{Linear Partial Differential Operator with constant coefficients  (LPDOcc)} acts linearly on functions $f \in \Ci(\R^n,\R)$, and by duality acts also on distributions. In what follows, each $a_{\alpha}$ will be an element of $\R$. By definition, only a finite number of such $a_{\alpha}$ are non-zero.
   
   \begin{equation}
   \label{eq:Dfundistr}
   D  : f \mapsto \left( z \mapsto \sum_{\alpha \in \N^n} a_{\alpha} \frac{\partial^{|\alpha|} f}{\partial x^{\alpha} }(z) \right) \qquad
   \hat{D}  : f \mapsto \left( z \mapsto \sum_{\alpha \in \N^n} (-1)^{|\alpha|} a_{\alpha} \frac{\partial^{|\alpha|} f}{\partial x^{\alpha} }(z)\right)
   \end{equation}
  
   \begin{rem} 
      The coefficients $(-1)^{|\alpha|}$ in equation \ref{eq:Dfundistr} originate from the intuition of distributions as generalised functions.
      With this intuition, it is natural to want that for each smooth function $f$, $D(T_f) = T_{D(f)}$, where $T_f$ stands for the distribution generalising the function $f$.
     When computing $T_{D(f)}$ on a function $g$ with integration by parts one shows that:
     \\ ${T_{D(f)}(g) = \int D(f) g = \int f (\hat{D}(g)) = (T_{f} \circ \hat{D})(g)}$, hence the definition.
     \end{rem}
  
   We make $D$ act on distributions through the following equation:
  
  \begin{equation}\label{eq:lpdodistrib}
      {D}(\phi) := \left( \phi \circ \hat{D} : f \mapsto \phi(\hat{D}(f)) \right) \in \Ci(\R^n,\R)'.
  \end{equation}
Notice the involutivity of $D\mapsto \hat{D}$, allowing us to state $\hat{D}(\phi) = \phi \circ D$.
   \begin{defi}
   Let $D$ be a LPDOcc. 
   A \emph{fundamental solution}
   of $D$ is a distribution $\Phi_D \in \C^{\infty}(\R^n,\R)'$ such that
   $ D(\Phi_D) = \delta_0. $
   \end{defi}
  
  \begin{prop}[Hormander, 1963]\label{prop:Ddistrconv}
   LPDOcc \emph{distribute over convolution}, meaning that $D(\phi \ast \psi) = D(\phi) \ast \psi = \phi \ast D(\psi)$ for any $\phi, \psi \in \oc E$. 
  \end{prop}
  The previous proposition is easy to check and means that knowing the fundamental solution of $D$ gives access to the solution $\psi \ast \Phi_D$ of the equation $D(\_)=\psi$. It is also the reason why indexation with several differential operators is possible. Luckily for us, LPDOcc are particularly well-behaved and always have a fundamental solution. The proof of the following well-known theorem can, for example, be found in~\cite[3.1.1]{hormander1963}.
   \begin{thm}[name = Malgrange-Ehrenpreis]\label{theorem:fondsol}
   Every linear partial differential operator with constant coefficients admits exactly one fundamental solution.
   \end{thm}
  Using this result, $\DDiLL$ gives new definitions for $\der$ and $\bd$, depending on a LPDOcc $D$:
  \begin{center}
  $
  \der_D : f \mapsto \Phi_D \ast f  \qquad \bd_{D}  : \phi \mapsto \phi \circ D .
  $
  \end{center}
  These new definitions came from the following ideas. Through the involutory duality, each $v \in E$ corresponds to a unique $\delta_v \in E ''\simeq E$, and $\bd_D$ is then interpreted as $ \phi \in  E''  \mapsto  \phi \circ D$. 
  
  While the interpretation of exponential \emph{rules} will not change, we will change the interpretation of exponential \emph{connectives} described by Kerjean for proof-theoretical reasons. We recall them now for comparison but will define new ones in Section~\ref{sec:idill}. $\DDiLL$ considered that $E'' = (D_0(?(E'),\R))'$ and made an analogy by replacing $D_0$ with~$D$, defining $?_{D} E := D(\Ci(E',\R))$. This gave types $\der_D : \wn_{D} E' \to \wn E'$ and $\bd_D : \oc_{D} E \to  \oc E $. Note that these definitions are sweeping reflexivity under the rug, and that no proof-theoretical constructions are given to account for the isomorphism $A \simeq A''$.
  
  The reader should note that these definitions only work for finite-dimensional vector spaces: one is able to apply a LPDOcc to a smooth function from $\R^n$ to $\R$ using partial differentiation on each dimension, but this is completely different if the function has an infinite-dimensional domain. The exponential connectives indexed by a LPDOcc therefore only apply to \emph{finitary} formulas: that is, the formulas with no exponentials.

  \subsection{Indexed linear logics: resources, effects and coeffects}
  
    Since Girard's original $\BLL$~\cite{bll}, several systems have implemented indexed exponentials to keep track of resource usage~\cite{bll-revisited,bll-cat}. More recently, several authors~\cite{GhicaSmith,effcoeff-grading,corequantitative} have defined a modular (but a bit less expressive) version $\BSLL$ where the exponentials are indexed (more specifically ``graded'', as in graded algebras) by elements of a given semiring $\mathcal{S}$.

  \begin{defi}\label{def:semiring}
      A \emph{semiring} $(\Sr, +, 0, \times, 1)$
      is given by a set $\mathcal{S}$ with two associative binary operations on $\mathcal{S}$: a sum $+$ which is commutative and has a neutral element $0\in\mathcal{S}$ and a product $\times$ which is distributive over the sum and has a neutral element $1\in\mathcal{S}$.
      \\ Such a semiring is said to be \emph{commutative} when the product is commutative.
      \\ An \emph{ordered semiring} is a semiring endowed with a partial order $\leq$ such that the sum and the product are monotonic.
  \end{defi}
  
  This type of indexation, named \emph{grading}, has been used in particular to study effects and coeffects, as well as resources~\cite{corequantitative,splitting,effcoeff-grading}.
  The main feature is to use this grading in a type system where some types are indexed by elements of the semiring. 
  This is exactly what is done in the logic $\BSLL$, where $\Sr$ is an ordered semiring.
  The exponential rules of $\BSLL$ are adapted from those of $\LL$, and agree with the intuitions that the index $x$ in $\oc_x A$ is a witness for the usage of resources of type $A$ by the proof/program.
   \vspace{\baselineskip}
  \begin{center}
      $\infer[\w]{\Gamma,\oc_0 A \vdash B}{\Gamma \vdash B}
  \qquad \infer[\contr]{\Gamma, \oc_{x + y} A \vdash B}{\Gamma,\oc_x A,\oc_y A \vdash B}
  \qquad \infer[\der]{\Gamma, \oc_1 A \vdash B}{\Gamma, A \vdash B}
  \qquad \infer[\p]{\oc_{x_1\times y}A_1,\dots,\oc_{x_n\times y}A_n \vdash \oc_y B}{\oc_{x_1}A_1,\dots,\oc_{x_n}A_n\vdash B}
  $
  \end{center}
   \vspace{\baselineskip}
  Finally, an indexed dereliction rule is also added, which uses the order of $\Sr$. In Section~\ref{sec:dbsll}, we will use an order induced by the additive rule of $\Sr$.
  \vspace{\baselineskip}
  \begin{center}
      $\infer[\der_I]{\Gamma,\oc_y A \vdash B}{\Gamma,\oc_x A \vdash B & x \leq y}$
  \end{center}
  \vspace{\baselineskip}
In graded linear logic, this rule is usually considered as a subtyping rule.
However, semantically in models of $\DiLL$, it corresponds to a variant of the dereliction, hence its name here.

  \section{A differential $\BSLL$}\label{sec:dbsll}
  In this section, we extend a graded linear logic with indexed coexponential rules. We define and prove correct a cut-elimination procedure.
  
  \subsection{Formulas and proofs}
  We define a differential version of $\BSLL$ by extending its set of exponential rules.
  The grammar for this logic is the same as the grammar of $\DiLL$, except that we add the graded exponentials:
    \[
      A,B := 0\mid 1 \mid \top \mid \bot \mid A\otimes B \mid A \parr B \mid A \with B \mid A \oplus B \mid \wn A \mid \oc A \mid \wn_x A \mid \oc_x A \qquad x \in \Sr.
    \]
  For the rules, we will restrict ourselves to a version without promotion, as it has been done for $\DiLL$ originally.
  Following the ideas behind $\DiLL$, we add \emph{costructural} exponential rules: a coweakening $\bw$, a cocontraction $\bc$, an indexed codereliction $\bdi$ and a codereliction $\bd$.
  The set of exponential rules of our new logic $\DBSLL$ is given in Figure \ref{fig:dbsllexp}. Note that by doing so, we study a \emph{classical} version of $\BSLL$, with an involutive linear duality.
  
  \begin{figure}
          \begin{center}
              \AXC{$\vdash \Gamma$}
              \RL{$\w$}
              \UIC{$\vdash \Gamma,\wn_0 A$}
              \DisplayProof
              \qquad
              \AXC{$\vdash \Gamma,\wn_x A,\wn_y A$}
              \RL{$\contr$}
              \UIC{$\vdash \Gamma,\wn_{x + y} A$}
              \DisplayProof
              \qquad
              \AXC{$\vdash \Gamma,\wn_x A$}
              \AXC{$x \leq y$}
              \RL{$\der_I$}
              \BIC{$\vdash \Gamma,\wn_y A$}
              \DisplayProof
              \qquad
              \AXC{$\vdash \Gamma, A$}
              \RL{$\der$}
              \UIC{$\vdash \Gamma, \wn A$}
              \DisplayProof
          \end{center}
          \begin{center}
              \AXC{}
              \RL{$\bw$}
              \UIC{$\vdash \oc_0 A$}
              \DisplayProof
              \qquad
              \AXC{$\vdash \Gamma, \oc_x A$}
              \AXC{$\vdash \Delta, \oc_y A$}
              \RL{$\bc$}
              \BIC{$\vdash \Gamma,\Delta, \oc_{x + y} A$}
              \DisplayProof
              \qquad
              \AXC{$\vdash \Gamma,\oc_x A$}
              \AXC{$x \leq y$}
              \RL{$\bdi$}
              \BIC{$\vdash \Gamma,\oc_y A$}
              \DisplayProof
              \qquad
              \AXC{$\vdash \Gamma, A$}
              \RL{$\bd$}
              \UIC{$\vdash \Gamma, \oc A$}
              \DisplayProof
          \end{center}
      \caption{Exponential rules of $\DBSLL$}\label{fig:dbsllexp}
  \end{figure}
  \begin{rem}
      In $\BSLL$, we consider a semiring $\mathcal{S}$ as a set of indices.
      With $\DBSLL$, we do not need a semiring: since this is a promotion-free version, only one operation (the sum) is important.
      Hence, in $\DBSLL$, $\mathcal{S}$ will only be a monoid. This modification requires two precision:
      \begin{itemize}
          \item The indexed (co)dereliction uses the fact that the elements of $\mathcal{S}$ can be compared through an order. Here, this order will \emph{always} be defined through the sum:
          $\forall x,y \in \mathcal{S},\ x \leq y \Longleftrightarrow \exists x' \in \mathcal{S},\ x+x' = y$. This is due to the fact that for compatibility with coexponential rules, we always need that each element of $\Sr$ is greater than 0.
          To be precise, this is sometimes only a \emph{preorder}, but it is not an issue in what follows.
          \item In $\BSLL$, the dereliction is indexed by $1$, the neutral element of the product. In $\DBSLL$, we will remove this index since we do not have a product operation and simply use~$!$ and~$?$ instead of~$!_1$ and~$?_1$. This implies that these non-indexed connectives, coming from $\der$ and $\bd$, will not interact with the indexed ones. However, we choose to keep these rules and these connectives in order to keep our syntax similar to that of graded linear logic.
      \end{itemize}
  \end{rem}
  Since every element of $\Sr$ is greater than 0, we have two admissible rules which will appear in the cut elimination procedure: an indexed weakening $\w_I$ and an indexed coweakening $\bwi$:
  \[
      \infer[\w_I]{\vdash \Gamma, \wn_x A}{\vdash \Gamma}\quad
      := \quad
      \begin{bpt}
          \AXC{$\vdash \Gamma$}
          \RL{$\w$}
          \UIC{$\vdash \Gamma, \wn_0 A$}
          \RL{$\der_I$}
          \UIC{$\vdash \Gamma, \wn_x A$}
      \end{bpt}
      \qquad \qquad \qquad
      \infer[\bwi]{\vdash \oc_x A}{}\quad
      :=\quad
      \begin{bpt}
          \AXC{}
          \RL{$\bw$}
          \UIC{$\vdash \oc_0 A$}
          \RL{$\bd_I$}
          \UIC{$\vdash \oc_x A$}
      \end{bpt}
  .\]

  \subsection{Definition of the cut elimination procedure}
  Since this work is done with a Curry-Howard perspective, a crucial point is the definition of a cut-elimination procedure.
  The cut rule is the following:
  \begin{center}
      \AXC{$\vdash \Gamma, A$}
      \AXC{$\vdash A^\bot, \Delta$}
      \RL{$cut$}
      \BIC{$\vdash \Gamma, \Delta$}
      \DisplayProof
  \end{center}
  which represents the composition of proofs/programs.
  Defining its elimination corresponds to expressing explicitly how to rewrite a proof with cuts into a proof without any cuts.
  It represents exactly the computations of our logic.
  
  In order to define the cut elimination procedure of $\DBSLL$, we have to consider the cases of cuts after each costructural rule that have been introduced, since
  the cases of cuts after $\MALL$ rules or after $\w,\, \contr,\, \der_I$ and $\der$ are already known.
  An important point is that we will use the formerly introduced indexed (co)weakening rather than the usual one.
  \par Before giving the formal rewriting of each case, we will divide them into three groups.
  Since $\DBSLL$ is highly inspired by $\DiLL$, one can try to adapt the cut-elimination procedure from $\DiLL$.
  This adaptation would mean that the structure of the rewriting is exactly the same, but the exponential connectives have to be indexed. 
  For most cases, this method works and there is exactly one possible way to index these connectives, since $\w_I$, $\bwi$, $\contr$, $\bc$, $\der$ and $\bd$ do not require a choice of the index (at this point, one can think that there is a choice in the indexing of $\w_I$ and $\bwi$, but this is a forced choice thanks to the other rules).
  \par However, the case of the cut between a contraction and a cocontraction will require some work on the indices because these two rules use the addition of the monoid. 
  The index of the principal formula $x$ (resp. $x'$) of a contraction (resp. cocontraction) rule is the sum of two indices $x_1$ and $x_2$ (resp. $x_3$ and $x_4$). 
  But $x = x'$ does not imply that $x_1 = x_3$ and~${x_2 = x_4}$. 
  We will then have to use a technical algebraic notion known as additive splitting to decorate the indices of the cut elimination between $\contr$ and $\bc$ in $\DiLL$.
  
  \begin{defi}\label{def:addsplit}
      A monoid $(\mathcal{M},+,0)$ is \emph{additive splitting} if for each $x_1,x_2,x_3,x_4 \in \mathcal{M}$ such that $x_1 + x_2 = x_3 + x_4$, 
      there are elements $x_{1,3},\ x_{1,4},\ x_{2,3},\ x_{2,4} \in \mathcal{M}$ such that 
      \[ x_1 = x_{1,3} + x_{1,4} \qquad x_2 = x_{2,3} + x_{2,4} \qquad x_3 = x_{1,3} + x_{2,3} \qquad x_4 = x_{1,4} + x_{2,4}.\]
  \end{defi}
  Graphically, this notion of splitting can be represented by the following diagram.
  
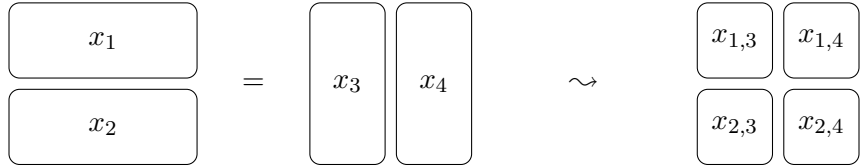
\begin{center}
    \begin{tikzpicture}[x=0.5cm,y=0.5cm]
    \filldraw[fill=gray!0!white, draw=black,rounded corners = 4] (0,0) rectangle (5,2);
    \filldraw[fill=gray!0!white, draw=black,rounded corners = 4] (0,2.3) rectangle (5,4.3);
    \node at (6.5,2.15) {=};
    \filldraw[fill=gray!0!white, draw=black,rounded corners = 4] (8,0) rectangle (10,4.3);
    \filldraw[fill=gray!0!white, draw=black,rounded corners = 4] (10.3,0) rectangle (12.3,4.3);
    \node at (15.3,2.15) {$\leadsto$};
    \filldraw[fill=gray!0!white, draw=black,rounded corners = 4] (18.3,0) rectangle (20.3,2);
    \filldraw[fill=gray!0!white, draw=black,rounded corners = 4] (18.3,2.3) rectangle (20.3,4.3);
    \filldraw[fill=gray!0!white, draw=black,rounded corners = 4] (20.6,0) rectangle (22.6,2);
    \filldraw[fill=gray!0!white, draw=black,rounded corners = 4] (20.6,2.3) rectangle (22.6,4.3);
    \node at (2.5,3.3) {$x_1$};
    \node at (2.5,1) {$x_2$};
    \node at (9,2.15) {$x_3$};
    \node at (11.3,2.15) {$x_4$};
    \node at (19.3,1) {$x_{2,3}$};
    \node at (19.3,3.3) {$x_{1,3}$};
    \node at (21.6,1) {$x_{2,4}$};
    \node at (21.6,3.3) {$x_{1,4}$};
    \end{tikzpicture}
    \end{center}
  This notion appears in~\cite{splitting}, for describing particular models of $\BSLL$, based on the relational model.
  Here, the purpose is different: it appears from a syntactical point of view. In the rest of this section, we will not only require $\mathcal{S}$ to be a monoid, but to be additive splitting as well. 
  
  \par Now that we have raised some fundamental differences in a possible cut-elimination procedure, one can note that we have not mentioned how to rewrite the cuts following an indexed (co)dereliction. 
  This is because the procedure from $\DiLL$ cannot be adapted at all in order to eliminate those cuts, as  $\der_I$ and $\bdi$ have nothing in common with the exponential rules of $\DiLL$.
  The situation is even worse: these cuts cannot be eliminated since these rules are not deterministic because of the use of the order relation.
  These considerations lead to the following division between the cut elimination cases.
  \begin{description}
      \item[Group 1] The cases where $\DiLL$ can naively be decorated. These will be cuts involving two exponential rules, with at least one being an indexed (co)weakening or a non-indexed (co)dereliction.
      \item[Group 2] The case where $\DiLL$ can be adapted using algebraic technicality, which is the cut between a contraction and a cocontraction.  
      \item[Group 3] The cases highly different from $\DiLL$. Those are the ones involving an indexed dereliction or an indexed codereliction.
  \end{description}
  The formal rewritings for the cases of groups 1 and 2 are given in Figure~\ref{fig:dbsllcut}. 
  The cut-elimination for contraction and a cocontraction uses the additive splitting property with the notations of Definition~\ref{def:addsplit}.
  \begin{figure}
  \begin{small}
      \begin{flalign*}
        &
       \begin{bpt}
        \alwaysNoLine
        \AXC{$\Pi_1$}
        \UIC{$\vdash \Gamma$}
        \alwaysSingleLine
        \RL{$\w_I$}
        \UIC{$\vdash \Gamma, \wn_x A$}
        \alwaysNoLine
        \AXC{}
        \alwaysSingleLine
        \RL{$\bwi$}
        \UIC{$\vdash \oc_x A^\bot$}
        \RL{$cut$}
        \BIC{$\vdash \Gamma$}
        \end{bpt} 
        &
        \lto_{cut} \quad
        &
        \begin{bpt}
        \alwaysNoLine
        \AXC{$\Pi_1$}
        \UIC{$\vdash \Gamma$}
        \end{bpt}
        & 
        &&\\
        &
        \begin{bpt}
            \alwaysNoLine
            \AXC{$\Pi_1$}
            \UIC{$\vdash \Gamma, A$}
            \alwaysSingleLine
            \RL{$\der$}
            \UIC{$\vdash \Gamma, \wn A$}
            \alwaysNoLine
            \AXC{$\Pi_2$}
            \UIC{$\vdash \Delta, A^\bot$}
            \alwaysSingleLine
            \RL{$\bd$}
            \UIC{$\vdash \Delta, \oc A^\bot$}
            \RL{$cut$}
            \BIC{$\vdash \Gamma, \Delta$}
            \end{bpt}
            &
            \lto_{cut} \quad
            &
            \begin{bpt}
            \alwaysNoLine
            \AXC{$\Pi_1$}
            \UIC{$\vdash \Gamma, A$}
            \AXC{$\Pi_2$}
            \UIC{$\vdash \Delta, A^\bot$}
            \alwaysSingleLine
            \RL{$cut$}
            \BIC{$\vdash \Gamma, \Delta$}
        \end{bpt}
        &
        &&\\
        &
        \begin{bpt}
        \alwaysNoLine
        \AXC{$\Pi_1$}
        \UIC{$\vdash \Gamma,\wn_x A,\wn_y A$}
        \alwaysSingleLine
        \RL{$\contr$}
        \UIC{$\vdash \Gamma,\wn_{x+y} A$}
        \alwaysNoLine
        \AXC{}
        \alwaysSingleLine
        \RL{$\bwi$}
        \UIC{$\vdash \oc_{x+y} A^\bot$}
        \RL{$cut$}
        \BIC{$\vdash \Gamma$}
        \end{bpt}
        &
        \lto_{cut} \quad
        &
        \begin{bpt}
        \alwaysNoLine
        \AXC{$\Pi_1$}
        \UIC{$\vdash \Gamma,\wn_x A,\wn_y A$}
        \AXC{}
        \alwaysSingleLine
        \RL{$\bwi$}
        \UIC{$\vdash \oc_y A^\bot$}
        \RL{$cut$}
        \BIC{$\vdash \Gamma,\wn_x A$}
        \alwaysNoLine
        \AXC{}
        \alwaysSingleLine
        \RL{$\bwi$}
        \UIC{$\vdash \oc_x A^\bot$}
        \RL{$cut$}
        \BIC{$\vdash \Gamma$}
        \end{bpt}
        &
        &&\\
        &
      \begin{bpt}
      \alwaysNoLine
      \AXC{$\Pi_1$}
      \UIC{$\vdash \Gamma,\oc_x A$}
      \AXC{$\Pi_2$}
      \UIC{$\vdash \Delta,\oc_y A$}
      \alwaysSingleLine
      \RL{$\bc$}
      \BIC{$\Gamma,\Delta,\oc_{x+y} A$}
      \alwaysNoLine
      \AXC{$\Pi_3$}
      \UIC{$\vdash \Xi$}
      \alwaysSingleLine
      \RL{$\w_I$}
      \UIC{$\vdash \Xi,\wn_{x+y} A^\bot$}
      \RL{$cut$}
      \BIC{$\vdash \Gamma,\Delta,\Xi$}
      \end{bpt}
      &
      \lto_{cut} \quad
      &
      \begin{bpt}
      \alwaysNoLine
      \AXC{$\Pi_1$}
      \UIC{$\vdash \Gamma,\oc_x A$}
      \AXC{$\Pi_3$}
      \UIC{$\vdash \Xi$}
      \alwaysSingleLine
      \RL{$\w_I$}
      \UIC{$\vdash \Xi,\wn_x A^\bot$}
      \RL{$cut$}
      \BIC{$\vdash \Gamma,\Xi$}
      \RL{$\w_I$}
      \UIC{$\vdash \Gamma,\Xi,\wn_y A^\bot$}
      \alwaysNoLine
      \AXC{$\Pi_2$}
      \UIC{$\vdash \Delta,\oc_y A$}
      \alwaysSingleLine
      \RL{$cut$}
      \BIC{$\vdash \Gamma,\Xi,\Delta$}
      \end{bpt}
      &
      &&\\
    \end{flalign*}
        
    \begin{flalign*}
    \begin{bpt}
          \alwaysNoLine
          \AXC{$\Pi_1$}
          \UIC{$\vdash \Gamma, \wn_{x_1} A^\bot, \wn_{x_2} A^\bot$}
          \alwaysSingleLine
          \RL{$\contr$}
          \UIC{$\vdash \Gamma,\wn_{x_1+x_2} A^\bot$}
          \alwaysNoLine
          \AXC{$\Pi_2$}
          \UIC{$\vdash \Delta,\oc_{x_3} A$}
          \AXC{$\Pi_3$}
          \UIC{$\vdash \Xi,\oc_{x_4} A$}
          \alwaysSingleLine
          \RL{$\bc$}
          \BIC{$\vdash \Delta,\Xi,\oc_{x_3+x_4 = x_1 + x_2} A$}
          \RL{$cut$}
          \BIC{$\vdash \Gamma,\Delta,\Xi $}
          \end{bpt}
          \lto_{cut} \quad 
          &&
        \end{flalign*}
            
        \begin{flalign*}
            &&
            \begin{bpt}
                \AXC{$\Pi_b$}
                \alwaysNoLine
                \UIC{$\vdash \Gamma, \wn_{x_{1,4}} A^\bot,\wn_{x_{2,4}} A^\bot, \wn_{x_3} A^\bot$}
                \AXC{$\Pi_2$}
                \alwaysNoLine
                \UIC{$\vdash \Delta,\oc_{x_3} A$}
                \alwaysSingleLine
                \RL{$cut$}
                \BIC{$\vdash \Gamma,\Delta,\wn_{x_{1,4}} A^\bot,\wn_{x_{2,4}} A^\bot$}
                \RL{$\contr$}
                \UIC{$\vdash \Gamma,\Delta,\wn_{x_4} A^\bot$}
                \AXC{$\Pi_3$}
                \alwaysNoLine
                \UIC{$\Xi,\oc_{x_4} A$}
                \alwaysSingleLine
                \RL{$cut$}
                \BIC{$\vdash \Gamma,\Delta,\Xi$}
                \end{bpt} 
        \end{flalign*}

          \begin{align*}
            &\text{\normalsize in which $\Pi_a$ and $\Pi_b$ are as follows:} & \hfill\\
            & \Pi_a = \qquad
            \begin{bpt}
            \AXC{}
            \RL{$ax$}
            \UIC{$\vdash \wn_{x_{2,3}} A^\bot, \oc_{x_{2,3}} A$}
            \AXC{}
            \RL{$ax$}
            \UIC{$\vdash \wn_{x_{2,4}} A^\bot, \oc_{x_{2,4}} A$}
            \RL{$\bc$}
            \BIC{$\vdash \wn_{x_{2,3}} A^\bot,\wn_{x_{2,4}} A^\bot,\oc_{x_2} A$}
            \AXC{$\Pi_1$}
            \alwaysNoLine
            \UIC{$\vdash \Gamma,\wn_{x_1} A^\bot,\wn_{x_2} A^\bot$}
            \alwaysSingleLine
            \RL{$cut$}
            \BIC{$\vdash \Gamma,\wn_{x_{2,3}} A^\bot,\wn_{x_{2,4}} A^\bot,\wn_{x_1} A^\bot$}
            \end{bpt} & \\
            & & \\
          &\Pi_b = \qquad
          \begin{bpt}
          \AXC{$\Pi_a$}
          \alwaysNoLine
          \UIC{$\vdash \Gamma,\wn_{x_{2,3}} A^\bot,\wn_{x_{2,4}} A^\bot,\wn_{x_1} A^\bot$}
          \alwaysSingleLine
          \AXC{}
          \RL{$ax$}
          \UIC{$\vdash \wn_{x_{1,3}} A^\bot,\oc_{x_{1,3}} A$}
          \AXC{}
          \RL{$ax$}
          \UIC{$\vdash \wn_{x_{1,4}} A^\bot,\oc_{x_{1,4}} A$}
          \RL{$\bc$}
          \BIC{$\wn_{x_{1,3}} A^\bot,\wn_{x_{1,4}} A^\bot,\oc_{x_1} A $}
          \RL{$cut$}
          \BIC{$\vdash \Gamma, \wn_{x_{2,3}} A^\bot,\wn_{x_{2,4}} A^\bot, \wn_{x_{1,3}} A^\bot,\wn_{x_{1,4}} A^\bot$}
          \RL{$\contr$}
          \UIC{$\vdash \Gamma, \wn_{x_{1,4}} A^\bot,\wn_{x_{2,4}} A^\bot, \wn_{x_3} A^\bot$}
          \end{bpt}
        \end{align*}
  \end{small}
  \caption{Cut elimination for $\DBSLL$: group 1 and group 2}
  \label{fig:dbsllcut}
  \end{figure}
  
  Finally, the last possible case of an occurrence of a cut in a proof is the one where $\der_I$ or $\bdi$ is applied before the cut: the group 3. The following definition introduces rewritings where these rules go up in the derivation tree, and which will be applied before the cut elimination procedure.
  This technique is inspired by subtyping ideas, which make sense since $\der_I$ is originally defined as a subtyping rule.
  \begin{defi}
      The rewriting procedures $\lto_{\der_I}$ and $\lto_{\bdi}$ are defined on proof trees of $\DBSLL$.
      \begin{enumerate}
      \item When $\der_I$ (resp. $\bdi$) is applied after a rule $r$ and $r$ is either from $\MALL$ (except the axiom) or $r$ is $\bwi$, $\bc$, $\bdi$ (resp. $\w_I$, $\contr$, $\der_I$), $\bd$ or $\der$, the rewriting $\lto_{\der_I,1}$ (resp. $\lto_{\bdi,1}$) exchanges~$r$ and~$\der_I$ (resp.~$\bdi$) which is possible since $r$ and $\der_I$ do not have the same principal formula.
      \item When $\der_I$ or $\bdi$ is applied after a (co)contraction, the rewriting is 
      \begin{center}
      \begin{bpt}
      \alwaysNoLine
      \AXC{$\Pi$}
      \UIC{$\vdash \Gamma,\wn_{x_1} A ,\wn_{x_2} A$}
      \alwaysSingleLine
      \RL{$\contr$}
      \UIC{$\vdash \Gamma,\wn_{x_1 + x_2} A$}
      \RL{$\der_I$}
      \UIC{$\vdash \Gamma, \wn_{x_1 + x_2 + x_3} A$}
      \end{bpt} 
      \qquad$\lto_{\der_I,2}$\qquad
      \begin{bpt}
      \alwaysNoLine
      \AXC{$\Pi$}
      \UIC{$\vdash \Gamma,\wn_{x_1} A ,\wn_{x_2} A$}
      \alwaysSingleLine
      \RL{$\contr$}
      \UIC{$\vdash \Gamma,\wn_{x_1 + x_2} A$}
      \RL{$\w_I$}
      \UIC{$\vdash \Gamma, \wn_{x_1 + x_2} A, \wn_{x_3} A$}
      \RL{$\contr$}
      \UIC{$\vdash \Gamma, \wn_{x_1 + x_2 + x_3} A$}
      \end{bpt}
      \end{center}
      \begin{center}
      \begin{bpt}
      \alwaysNoLine
      \AXC{$\Pi_1$}
      \UIC{$\vdash \Gamma,\oc_{x_1} A$}
      \AXC{$\Pi_2$}
      \UIC{$\vdash \Delta,\oc_{x_2} A$}
      \alwaysSingleLine
      \RL{$\bc$}
      \BIC{$\vdash \Gamma,\Delta,\oc_{x_1 + x_2} A$}
      \RL{$\bdi$}
      \UIC{$\vdash \Gamma,\Delta,\oc_{x_1 + x_2 + x_3} A$}
      \end{bpt}
      \quad$\lto_{\bdi,2}$\quad
      \begin{bpt}
      \alwaysNoLine
      \AXC{$\Pi_1$}
      \UIC{$\vdash \Gamma,\oc_{x_1} A$}
      \alwaysNoLine
      \AXC{$\Pi_2$}
      \UIC{$\vdash \Delta,\oc_{x_2} A$}
      \alwaysSingleLine
      \RL{$\bc$}
      \BIC{$\vdash \Gamma,\Delta,\oc_{x_1 + x_2} A$}
      \AXC{$\vdash$}
      \RL{$\bwi$}
      \UIC{$\vdash \oc_{x_3} A$}
      \RL{$\bc$}
      \BIC{$\vdash \Gamma,\Delta,\oc_{x_1 + x_2 + x_3} A$}
      \end{bpt}
      \end{center}
      \item If it is applied after an indexed (co)weakening, the rewriting is 
      \begin{center}
      \begin{bpt}
      \alwaysNoLine
      \AXC{$\Pi$}
      \UIC{$\vdash \Gamma$}
      \alwaysSingleLine
      \RL{$\w_I$}
      \UIC{$\vdash \Gamma,\wn_x A$}
      \RL{$\der_I$}
      \UIC{$\vdash \Gamma,\wn_{x+y} A$}
      \end{bpt}
      $\lto_{\der_I,3}$
      \begin{bpt}
      \alwaysNoLine
      \AXC{$\Pi$}
      \UIC{$\vdash \Gamma$}
      \alwaysSingleLine
      \RL{$\w_I$}
      \UIC{$\vdash \Gamma,\wn_{x+y} A$}
      \end{bpt}
      \qquad
      \begin{bpt}
      \alwaysNoLine
      \AXC{$\Pi$}
      \UIC{$\vdash$}
      \alwaysSingleLine
      \RL{$\bwi$}
      \UIC{$\vdash \oc_x A$}
      \RL{$\bdi$}
      \UIC{$\vdash \oc_{x+y} A$}
      \end{bpt}
      $\lto_{\bdi,3}$
      \begin{bpt}
      \alwaysNoLine
      \AXC{$\Pi$}
      \UIC{$\vdash$}
      \alwaysSingleLine
      \RL{$\bwi$}
      \UIC{$\vdash \oc_{x+y} A$}
      \end{bpt}
      \end{center}
      \item And if it is after an axiom, we define
      \begin{center}
      \begin{bpt}
      \AXC{}
      \RL{$ax$}
      \UIC{$\vdash \oc_{x} A, \wn_{x} A^\bot $}
      \RL{$\der_I$}
      \UIC{$\vdash \oc_{x} A, \wn_{x+y} A^\bot $}
      \end{bpt}
      $\lto_{\der_I,4}$
      \begin{bpt}
      \AXC{}
      \RL{$ax$}
      \UIC{$\vdash \oc_{x} A,\wn_{x} A^\bot$}
      \RL{$\w_I$}
      \UIC{$\vdash \oc_{x} A,\wn_{x} A^\bot,\wn_{y} A^\bot$}
      \RL{$\contr$}
      \UIC{$\vdash \oc_{x} A,\wn_{x+y} A^\bot $}
      \end{bpt}
      \end{center}
      \begin{center}
      \begin{bpt}
      \AXC{}
      \RL{$ax$}
      \UIC{$\vdash \oc_x A,\wn_x A^\bot$}
      \RL{$\bdi$}
      \UIC{$\vdash \oc_{x+y} A,\wn_x A^\bot$}
      \end{bpt}
      $\lto_{\bdi,4}$
      \begin{bpt}
      \AXC{}
      \RL{$ax$}
      \UIC{$\vdash \oc_x A,\wn_x A^\bot$}
      \AXC{$\vdash$}
      \RL{$\bwi$}
      \UIC{$\vdash \oc_y A$}
      \RL{$\bc$}
      \BIC{$\vdash \oc_{x+y} A, \wn_x A^\bot$}
      \end{bpt}
      \end{center}
      \end{enumerate}
      One defines  $\lto_{\der_I}$ (resp. $\lto_{\bdi}$) as the transitive closure of the union of the $\lto_{\der_I,i}$ (resp. $\lto_{\bdi, i}$).
  \end{defi}
  
  Even if this definition is non-deterministic, this is not a problem.
  Every indexed (co)dereliction goes up in the tree, without meeting another one. This implies that this rewriting is confluent: the result of the rewriting does not depend on the choices made.
  \begin{rem}\label{rem:foncteur idill-dill}
      It is easy to define a forgetful functor $U$, which transforms a formula (resp. a proof) of $\DBSLL$ into a formula (resp. a proof) of $\DiLL$. 
      For a formula $A$ of $\DBSLL$, $U(A)$ is $A$ where each $\oc_x$ (resp. $\wn_x$) is transformed into $\oc$ (resp. $\wn$), which is a formula of $\DiLL$. 
      For a proof-tree without any $\der_I$ and $\bdi$, the idea is the same: when an exponential rule of $\DBSLL$ is applied in a proof-tree $\Pi$, the same rule but not indexed is applied in $U(\Pi)$, which is a proof-tree in $\DiLL$. 
      Moreover, we notice that if $\Pi_1 \lto_{cut} \Pi_2$, $U(\Pi_1) \lto_{\DiLL} U(\Pi_2)$ where $\lto_{\DiLL}$ is the cut-elimination in~\cite{ehrhardpn}.
  \end{rem}
  We can now define a cut-elimination procedure:
  \begin{defi}
      The rewriting $\lto$ is defined on derivation trees. 
      For a tree~$\Pi$, we apply~$\lto_{\der_I}$,~$\lto_{\bdi}$ and~$\lto_{cut}$ as long as it is possible. 
      When there are no more cuts, the rewriting ends.
  \end{defi}
  
  \begin{thm}
  \label{thm:cutelim}
      The rewriting procedure $\lto$ terminates on each derivation tree, and reaches an equivalent tree with no cut.
  \end{thm}
  
  In order to prove this theorem, we first need to prove a lemma, which shows that the (co)dereliction elimination is well-defined.
  
  \begin{lem}\label{lem:derelim}
  For each derivation tree $\Pi$, if we apply $\lto_{\der_I}$ and $\lto_{\bdi}$ to $\Pi$, this procedure terminates such that $\Pi \lto_{\der_I} \Pi_1 \lto_{\bdi} \Pi_2$ without any $\der_I$ and $\bdi$ in $\Pi_2$.
  \end{lem}
  \begin{proof}
      Let $\Pi$ be a proof-tree. 
      Each rule has a height (using the usual definition for nodes in a tree). 
      We define the \textit{depth} of a node as the height of the tree minus the height of this node. 
      The procedure $\lto_{\der_I}$ terminates on $\Pi$: let $a(\Pi)$ be the number of indexed derelictions in $\Pi$ and $b(\Pi)$ be the sum of the depth of each indexed dereliction in $\Pi$. Now, we define~${H(\Pi) = (a(\Pi), b(\Pi))}$ and $<_{lex}$ as the lexicographical order on $\N^2$. For each step of $\lto_{\der_I}$ such that $\Pi_i \lto_{\der_I} \Pi_j$, we have $H(\Pi_i) <_{lex} H(\Pi_j)$: 
      \begin{enumerate}
      \item If $\Pi_i \lto_{\der_I,1} \Pi_j$, the number of $\der_I$ does not change, and the sum of depths decreases by 1. Hence, $H(\Pi_i) <_{lex} H(\Pi_j)$.
      \item If $\Pi_i \lto_{\der_I,k} \Pi_j$ with $ 2 \leq k \leq 4$, the number of derelictions decreases, so $H(\Pi_i) <_{lex} H(\Pi_j)$.
      \end{enumerate}
      Using this property and the fact that $<_{lex}$ is a well-founded order on $\N^2$, this rewriting procedure has to terminate on a tree $\Pi_1$. 
      Moreover, if there is an indexed dereliction in $\Pi_1$, this dereliction is below another rule, so $\lto_{\der_I,i}$ for $1 \leq i \leq 4$ can be applied, which leads to a contradiction with the definition of $\Pi_1$. Then, there is no indexed dereliction in $\Pi_1$.
      \par Using similar arguments, the rewriting procedure $\lto_{\bdi}$ on $\Pi_1$ ends on a tree $\Pi_2$ where there is no codereliction (and no dereliction because the procedure $\lto_{\bdi}$ does not introduce any derelictions).
  \end{proof}
  
   \begin{proof}[Proof of Theorem \ref{thm:cutelim}]
      If we apply our procedure $\lto$ on a tree $\Pi$ we will, using Lemma~\ref{lem:derelim}, have a tree $\Pi_{\der_I,\bdi}$ such that $\Pi \lto_{\der_I} \Pi_{\der_I} \lto_{\bdi} \Pi_{\der_I,\bdi}$ and there is no indexed dereliction and no indexed codereliction in $\Pi_{\der_I,\bdi}$. 
      Hence, the procedure $\lto$ applied on $\Pi$ gives a rewriting
      \[ \Pi \lto_{\der_I} \Pi_{\der_I} \lto_{\bdi} \left(\Pi_{\der_I,\bdi} = \Pi_{0}\right) \lto_{cut} \Pi_{1} \lto_{cut} \dots \]
      Applying the forgetful functor $U$ from Remark \ref{rem:foncteur idill-dill} on each tree $\Pi_i$ (for $i \in \N$), the cut-elimination theorem of $\DiLL$~\cite{pagani_cut-elimination_2009} implies that this rewriting terminates at step $n$, because the cut-elimination rules of $\DBSLL$ which are used in $\Pi_0$ are those of $\DiLL$ when the indexes are removed.
      Then, $\Pi \lto^* \Pi_n$ where $\Pi_n$ is cut-free. 
   \end{proof}
  \begin{rem}
    Notice that while $\DiLL$ is famous for introducing \emph{formal sums of proofs} with its cut-elimination, we have none of that here.  Sums are generated by cut-elimination between $\bc$ and $\der$ or $\contr$ and $\bd$, mimicking the calculus rule for differentiation. LPDOcc do not behave like this, and fundamental solutions or differential operators are painlessly propagated into the first argument of a distribution or function.
    
    As far as syntax is concerned, we are only treating a weakened version of the (co)derelic\-tion, which is responsible for the sum in $\DiLL$. In a way, the labels, by allowing finer insight into the resource allocation, may remove or/and add such sums :
    \begin{itemize}
    \item Using positivity (i.e. the fact that $x+y = 0$ implies that $x = 0$ or $y = 0$), we could define a cut between a codereliction graded by $1$ and a contraction deterministically.
    \item Conversely, even though the additive splitting of LPDOcc, our example of interest, happens to be deterministic (see Section~\ref{sec:idill}), it is not always the case, and one may want to perform all possible choices non-deterministically, hence a new sum.
    \end{itemize}

  \end{rem}

  \subsection{The promotion rule}

  In $\DBSLL$, we do not consider the promotion rule. 
  However, this rule is crucial in programming languages' semantics, since it allows for representing higher-order programs.

  In the previous subsection, we have restricted $\Sr$ to be only a monoid, since without a promotion rule, the product operation was not useful.
  We will here study how we could add a promotion rule in $\DBSLL$, so we will consider a semiring $(\Sr, 0,+,1,\times)$.
  In $\BSLL$, the promotion rule is
  \[
      \infer[\p]{\vdash \wn_{x_1\times y}A_1,\dots,\wn_{x_n\times y}A_n, \oc_y B}{\vdash \wn_{x_1}A_1,\dots,\wn_{x_n}A_n, B}  
  \]
  Note that one has to be careful with the indices while using this rule, since the product is not necessarily commutative.
  If one wants to add this rule to $\DBSLL$, it has to extend the cut elimination procedure.
  The cases where the promotion interacts with the graded structural rules ($\w, \contr, \der, \der_I$ and itself) are studied in~\cite{splitting}.
  For the sake of completeness, we present these rules in Figure~\ref{fig:promgradcut}.
  \begin{figure}
    \begin{small}
         \begin{smallreduc}{
          \alwaysNoLine
          \AXC{$\Pi_1$}
          \UIC{$\vdash\wn_{y_1}A_1,\dots,\wn_{y_n}A_n,B$}
          \alwaysSingleLine
          \RL{$\p$}
          \UIC{$\vdash\wn_{0\times y_1}A_1,\dots,\wn_{0\times y_n}A_n,\oc_0 B$}
          \alwaysNoLine
          \AXC{$\Pi_2$}
          \UIC{$\vdash\Gamma$}
          \alwaysSingleLine
          \RL{$\w$}
          \UIC{$\vdash\Gamma,\wn_0 B^\bot$}
          \RL{$cut$}
          \BIC{$\vdash\Gamma,\wn_{0}A_1,\dots,\wn_{0}A_n$}
          \DisplayProof{}} 
          {\alwaysNoLine
          \AXC{$\Pi_2$}
          \UIC{$\vdash \Gamma$}
          \alwaysSingleLine
          \RL{$\w$}
          \UIC{$\vdash\Gamma,\wn_0 A_1$}
          \RL{$\w$}
          \UIC{$\vdots$}
          \RL{$\w$}
          \UIC{$\vdash\Gamma,\wn_0 A_1,\dots,\wn_0 A_n$}
          \DisplayProof{}
          }
         \end{smallreduc}
          \begin{bigreduc}
          {
        \alwaysNoLine
          \AXC{$\Pi_1$}
          \UIC{$\vdash\wn_{z_1}A_1,\dots,\wn_{z_n}A_n,B$}
          \alwaysSingleLine
          \RL{$\p$}
          \UIC{$\vdash\wn_{(x+y)z_1}A_1,\dots,\wn_{(x+y)z_n}A_n,\oc_{x+y} B$}
            \alwaysNoLine
            \AXC{$\Pi_2$}
            \UIC{$\vdash \Gamma, \wn_x B^\bot,\wn_y B^\bot$}
            \alwaysSingleLine
            \RL{$\contr$}
            \UIC{$\vdash \Gamma, \wn_{x+y} B^\bot$}
            \RL{$cut$}
            \BIC{$\vdash \Gamma,\wn_{xz_1+yz_1}A_1,\dots,\wn_{xz_1+yz_n}A_n$}
            \DisplayProof{}
            }
            {
            \alwaysNoLine
            \AXC{$\Pi_1$}
            \UIC{$\vdash\wn_{z_1}A_1,\dots,\wn_{z_n}A_n,B$}
            \alwaysSingleLine
            \RL{$\p$}
            \UIC{$\vdash\wn_{xz_1}A_1,\dots,\wn_{xz_n}A_n,\oc_{x} B$}
            \alwaysNoLine
            \AXC{$\Pi_1$}
            \UIC{$\vdash\wn_{z_1}A_1,\dots,\wn_{z_n}A_n,B$}
            \alwaysSingleLine
            \RL{$\p$}
            \UIC{$\vdash\wn_{yz_1}A_1,\dots,\wn_{yz_n}A_n,\oc_{y} B$}
            \alwaysNoLine
            \AXC{$\Pi_2$}
            \UIC{$\vdash \Gamma, \wn_x B^\bot,\wn_y B^\bot$}
            \alwaysSingleLine
            \RL{$cut$}
            \BIC{$\vdash \Gamma, \wn_{yz_1}A_1,\dots,\wn_{yz_n}A_n,\wn_{x} B^\bot$}
            \RL{$cut$}
            \BIC{$\vdash \Gamma,\wn_{xz_1}A_1,\dots,\wn_{xz_n}A_n,\wn_{yz_1}A_1,\dots,\wn_{yz_n}A_n$}
            \RL{$\contr$}
            \UIC{$\vdots$}
            \RL{$\contr$}
            \UIC{$\vdash \Gamma,\wn_{xz_1+yz_1}A_1,\dots,\wn_{xz_1+yz_n}A_n$}
            \DisplayProof{}
            }
          \end{bigreduc}
          \begin{bigreduc}
            {
            \alwaysNoLine
            \AXC{$\Pi_1$}
            \UIC{$\vdash\wn_{z_1}A_1,\dots,\wn_{z_n}A_n,B^\bot$}
            \alwaysSingleLine
            \RL{$\p$}
            \UIC{$\vdash\wn_{xyz_1}A_1,\dots,\wn_{xyz_n}A_n,\oc_{xy}B^\bot$}
            \alwaysNoLine
            \AXC{$\Pi_2$}
            \UIC{$\vdash\wn_{y_1}B_1,\dots,\wn_{y_m}B_m,\wn_{y}B^\bot,C$}
            \alwaysSingleLine
            \RL{$\p$}
            \UIC{$\vdash\wn_{xy_1}B_1,\dots,\wn_{xy_m}B_m,\wn_{xy}B^\bot,\oc_x C$}
            \RL{$cut$}
            \BIC{$\vdash\wn_{xyz_1}A_1,\dots,\wn_{xyz_n}A_n,\wn_{xy_1}B_1,\dots,\wn_{xy_m}B_m,\oc_x C$}
            \DisplayProof{}
            }
            {
            \alwaysNoLine
            \AXC{$\Pi_1$}
            \UIC{$\vdash\wn_{z_1}A_1,\dots,\wn_{z_n}A_n,B^\bot$}
            \alwaysSingleLine
            \RL{$\p$}
            \UIC{$\vdash\wn_{yz_1}A_1,\dots,\wn_{yz_n}A_n,\oc_{y}B^\bot$}
            \alwaysNoLine
            \AXC{$\Pi_2$}
            \UIC{$\vdash\wn_{y_1}B_1,\dots,\wn_{y_m}B_m,\wn_{y}B^\bot,C$}
            \alwaysSingleLine
            \RL{$cut$}
            \BIC{$\vdash\wn_{yz_1}A_1,\dots,\wn_{yz_n}A_n,\wn_{y_1}B_1,\dots,\wn_{y_m}B_m,C$}
            \RL{$\p$}
            \UIC{$\vdash\wn_{xyz_1}A_1,\dots,\wn_{xyz_n}A_n,\wn_{xy_1}B_1,\dots,\wn_{xy_m}B_m,\oc_x C$}
            \DisplayProof{}
            }
          \end{bigreduc}
          \begin{bigreduc}
            {
            \alwaysNoLine
            \AXC{$\Pi_1$}
            \UIC{$\vdash\wn_{x_1}A_1,\dots,\wn_{x_n}A_n,B$}
            \alwaysSingleLine
            \RL{$\p$}
            \UIC{$\vdash\wn_{(y+z)x_1}A_1,\dots,\wn_{(y+z)x_n}A_n,\oc_{y+z} B$}
            \alwaysNoLine
            \AXC{$\Pi_2$}
            \UIC{$\vdash \Gamma,\wn_y B^\bot$}
            \alwaysSingleLine
            \RL{$\der_I$}
            \UIC{$\vdash\Gamma,\wn_{y+z}B^\bot$}
            \RL{$cut$}
            \BIC{$\vdash\Gamma,\wn_{yx_1+zx_1}A_1,\dots,\wn_{yx_n+zx_n}A_n$}
            \DisplayProof{}
            }
            {
            \alwaysNoLine
            \AXC{$\Pi_1$}
            \UIC{$\vdash\wn_{x_1}A_1,\dots,\wn_{x_n}A_n,B$}
            \alwaysSingleLine
            \RL{$\p$}
            \UIC{$\vdash\wn_{yx_1}A_1,\dots,\wn_{yx_n}A_n,\oc_{y} B$}
            \alwaysNoLine
            \AXC{$\Pi_2$}
            \UIC{$\vdash \Gamma,\wn_y B^\bot$}
            \alwaysSingleLine
            \RL{$cut$}
            \BIC{$\vdash\Gamma,\wn_{yx_1}A_1,\dots,\wn_{yx_n}A_n$}
            \RL{$\der_I$}
            \UIC{$\vdots$}
            \RL{$\der_I$}
            \UIC{$\vdash\Gamma,\wn_{yx_1+zx_1}A_1,\dots,\wn_{yx_n+zx_n}A_n$}
            \DisplayProof{}
            }
          \end{bigreduc}
    \end{small}
    \caption{Cut elimination for the promotion rule in graded linear logic}
    \label{fig:promgradcut}
    \end{figure}
  Here we describe how to eliminate the cuts between a costructural rule and a promotion.
  To do so, we will need some additional properties on the semiring $\Sr$.
  \begin{defi}\label{mult_split}
    Let $(\Sr, 0,+,1,\times)$ be a semiring.
    \begin{itemize}
        \item $\Sr$ is an \emph{integral domain} if for each non-zero elements $x,y$, $xy \neq 0$.
        \item $\Sr$ is \emph{multiplicative splitting} when, if $sr = x+y$, there are elements $s_1,\dots,s_n,r_1,\dots,r_m \in \Sr$ and a set $U \subseteq \{1,\dots,n\} \times \{1,\dots,m\}$ such that 
        \[  s = \sum_{i=1}^n s_i \qquad r = \sum_{j=1}^m r_j \qquad x = \sum_{(i,j) \in U}s_i r_j \qquad y = \sum_{(i,j) \notin U}s_i r_j. \]
    \end{itemize}
    \end{defi}
In what follows, we will assume that $\Sr$ is both an integral domain and multiplicative splitting.
We can now give the rewriting cases of the cut elimination procedure with a promotion rule.
  \begin{itemize}
      \item The coweakening:
      A cut between a coweakening and a promotion is
      \[
          \begin{bpt}
              \AXC{}
              \RL{$\bw$}
              \UIC{$\vdash \oc_0 A$}
              \AXC{$\Pi$}
              \alwaysNoLine
              \UIC{$\vdash \wn_x A^\bot, \wn_{y_1}B_1,\dots \wn_{y_n}B_n,C$}
              \alwaysSingleLine
              \RL{$\p$}
              \UIC{$\vdash \wn_{xz} A^\bot, \wn_{y_1z}B_1,\dots \wn_{y_nz}B_n,\oc_z C$}
              \RL{$cut$}
              \BIC{$\vdash \wn_{y_1z}B_1,\dots \wn_{y_nz}B_n,\oc_z C$}
          \end{bpt} 
      \]
  with $xz = 0$. Since we have supposed that $\Sr$ is an integral domain, we have $x = 0$ or $z = 0$. Depending on whether $x$ or $z$ is equal to 0, the rewriting will not be the same.
  \par If $x=0$, the previous proof tree is rewritten as
      \[
      \begin{bpt}
          \AXC{}
          \RL{$\bw$}
          \UIC{$\vdash \oc_0 A$}
          \AXC{$\Pi$}
          \alwaysNoLine
          \UIC{$\vdash \wn_{x=0} A^\bot, \wn_{y_1}B_1,\dots \wn_{y_n}B_n,C$}
          \alwaysSingleLine
          \RL{$cut$}
          \BIC{$\vdash \wn_{y_1}B_1,\dots \wn_{y_n}B_n,C$}
          \RL{$\p$}
          \UIC{$\vdash \wn_{y_1z}B_1,\dots \wn_{y_nz}B_n,\oc_z C$}
      \end{bpt} 
      \]
      If $x \neq 0$, this rewriting does not work, since it is impossible to make a cut between $\Pi$ and a coweakening.
      However, $z=0$, so each index in the conclusion of the tree is equal to 0.
      We can then rewrite the tree as
      \[
          \begin{bpt}
          \AXC{}
          \RL{$\bw$}
          \UIC{$\vdash \oc_0 C$}
          \RL{$\w$}
          \UIC{$\vdash \wn_0 B_1, \oc_0 C$}
          \RL{$\w$}
          \UIC{$\vdots$}
          \RL{$\w$}
          \UIC{$\vdash \wn_0 B_1,\dots,\wn_0 B_n,\oc_0 C$}
          \end{bpt}
      \]
      where, after a first coweakening which introduces $\oc_0 C$, we do exactly $n$ weakening in order to introduce each $\wn_0 B_i$ for $1 \leq i \leq n$.
      \item The cocontraction: A cut between a cocontraction and a promotion is
      \[
          \begin{bpt}
              \alwaysNoLine
              \AXC{$\Pi_1$}
              \UIC{$\vdash \Gamma,\oc_x A$}
              \AXC{$\Pi_2$}
              \UIC{$\vdash \Delta,\oc_y A$}
              \alwaysSingleLine
              \RL{$\bc$}
              \BIC{$\vdash \Gamma, \Delta, \oc_{x+y}A$}
              \alwaysNoLine
              \AXC{$\Pi_3$}
              \UIC{$\vdash \wn_r A^\bot, \wn_{s}B,C$}
              \alwaysSingleLine
              \RL{$\p$}
              \UIC{$\vdash \wn_{rt} A^\bot, \wn_{st}B,\oc_t C$}
              \RL{$cut$}
              \BIC{$\vdash \Gamma,\Delta,\wn_{st}B,\oc_t C$}
          \end{bpt}
      \]
      with $x+y = rt$. Note that to lighten the notations, we have reduced the context to one formula $\wn_s B$, but this simplification does not change the way the rewriting works. Using the multiplicative splitting property of $\Sr$, there are elements $r_1,\dots r_k,t_1,\dots,t_l \in \Sr$, and a set $U \subseteq \{1,\dots,k\} \times \{1,\dots,l\}$ such that
      \[
          r = \sum_{i=1}^k r_i \qquad t = \sum_{j=1}^l t_j \qquad x = \sum_{(i,j) \in U}r_it_j \qquad y = \sum_{(i,j) \notin U}r_it_j.
      \]
      Before giving the rewriting of this case, we define a rule $\bc^\bot$ by 
      \[
          \begin{bpt}
              \alwaysNoLine
              \AXC{$\Pi$}
              \UIC{$\vdash \Gamma, \wn_{x+y}A$}
              \alwaysSingleLine
              \RL{$\bc^\bot$}
              \UIC{$\vdash \Gamma, \wn_x A, \wn_y A$}
          \end{bpt}
          \quad :=
          \quad
          \begin{bpt}
              \alwaysNoLine
              \AXC{$\Pi$}
              \UIC{$\vdash \Gamma, \wn_{x+y}A$}
              \alwaysSingleLine
              \AXC{}
              \RL{$ax$}
              \UIC{$\vdash \oc_x A^\bot, \wn_x A$}
              \AXC{}
              \RL{$ax$}
              \UIC{$\vdash \oc_y A^\bot, \wn_y A$}
              \RL{$\bc$}
              \BIC{$\vdash \oc_{x+y}A^\bot,\wn_x A,\wn_y A$}
              \RL{$cut$}
              \BIC{$\vdash \Gamma,\wn_x A,\wn_y A$}
          \end{bpt}
      \]
      which can be understood as the dual of the cocontraction rule. 
      One can note that this technique is used in the rewriting of a cut between a contraction and a cocontraction.
      From this, we define subtrees $\Pi_{3,j}$ for each $1\leq j \leq l$ as 
      \[  \Pi_{3,j} := \quad
          \begin{bpt}
              \alwaysNoLine
              \AXC{$\Pi_3$}
              \UIC{$\vdash \wn_{\sum_{i=1}^k r_i} A^\bot, \wn_S B,C$}
              \alwaysSingleLine
              \RL{$\bc^\bot$}
              \UIC{$\vdash \wn_{\sum_{\{i \mid (i,j) \in U\}}r_i}  A^\bot,\wn_{\sum_{\{i \mid (i,j) \notin U\}}r_i}  A^\bot, \wn_s B,C$}
              \RL{$\p$}
              \UIC{$\vdash \wn_{\sum_{\{i \mid (i,j) \in U\}}r_it_j}  A^\bot,\wn_{\sum_{\{i \mid (i,j) \notin U\}}r_it_j}  A^\bot, \wn_{st_j} B,\oc_{t_j}C$}
          \end{bpt}
      \]
      Where we use $\contr^\bot$ to split the sum $r = \sum_{i=1}^k r_i$ into two sums: the elements $r_i$ such that $r_it_j$ is in the decomposition of $x$, and the others (which are then in the decomposition of $y$).
      \par Now, we need one last intermediate step, before giving the rewriting. That is defining the subtree $\Pi_3'$, which will \emph{combine} each $\Pi_{3,j}$ using cocontractions:
      \[  \Pi_3' :=
          \begin{bpt}
              \alwaysNoLine
              \AXC{$\Pi_{3,1}$}
              \UIC{$\vdash \wn_{\sum_{(i,1) \in U}r_it_1}  A^\bot,\wn_{\sum_{(i,1) \notin U}r_it_1}  A^\bot, \wn_{st_1} B,\oc_{t_1}C$}
              \alwaysSingleLine
              \AXC{$\Pi_{3,2}$}
              \RL{$\bc$}
              \BIC{$\ddots$}
              \AXC{$\Pi_{3,l}$}
              \alwaysSingleLine
              \RL{$\bc$}
              \BIC{$\vdash \wn_{\sum_{(i,1) \in U}r_it_1}  A^\bot,\wn_{\sum_{(i,1) \notin U}r_it_1}  A^\bot,\dots,\wn_{st_1} B,\dots,\wn_{st_l}B,\oc_t C$}
              \RL{$\contr$}
              \UIC{$\vdots$}
              \RL{$\contr$}
              \UIC{$\vdash \wn_{\sum_{(i,j)\in U}r_it_j}A^\bot,\wn_{\sum_{(i,j)\notin U}r_it_j}A^\bot, \wn_{st}B,\oc_t C$}
          \end{bpt}
      \]
      Here, we have used several contractions in order to \emph{recombine} some formulas. Each $\wn_{\sum_{(i,j)\in U}r_i t_j} A^\bot$ has been contracted, and the index is now 
      \[ \left(\sum_{(i,1) \in U}r_it_1\right) + \dots + \left(\sum_{(i,l) \in U}r_it_l\right) = \sum_{(i,j) \in U}r_it_j = x. \]
      This is similar for the $\wn_{\sum_{(i,j)\notin U}r_i t_j} A^\bot$, and the final index is $y$.
      \par Finally, the rewriting is
      \[
          \begin{bpt}
              \alwaysNoLine
              \AXC{$\Pi_1$}
              \UIC{$\vdash \Gamma,\oc_x A$}
              \AXC{$\Pi_2$}
              \UIC{$\vdash \Delta,\oc_y A$}
              \AXC{$\Pi_3'$}
              \UIC{$\vdash \wn_x A^\bot, \wn_y A^\bot, \wn_{s_1t}B_1,\dots,\wn_{s_nt}B_n,\oc_t C$}
              \alwaysSingleLine
              \RL{$cut$}
              \BIC{$\vdash \Delta, \wn_x A^\bot,\wn_{s_1t}B_1,\dots,\wn_{s_nt}B_n,\oc_t C$}
              \RL{$cut$}
              \BIC{$\vdash \Gamma,\Delta,\wn_{s_1t}B_1,\dots,\wn_{s_nt}B_n,\oc_t C$}
          \end{bpt}
      \]
  \end{itemize}
  As for the promotion-free version of $\DBSLL$, we have not considered the cut elimination with an indexed (co)dereliction.
  In the previous subsection, this question was solved using a technique where these rules go up in the tree, which allows us to not consider these cases.
  In order to incorporate a promotion rule in $\DBSLL$, these indexed (co)derelictions should also commute with the promotion, if one wants to have a cut elimination procedure.
  \par Here, we face some issues.
  First, we do not know how to make these rules commute directly.
  Since the promotion involves the product, and the indexed (co)dereliction involves the order, it seems to require some algebraic properties.
  However, even with the definition of the order through the sum, and using the multiplicative splitting, we do not get any relation in the indices that would help us define a commutation.
  A natural idea to solve this issue would be to adapt what we have done for the (co)weakening: define a rule which combines a promotion and an indexed (co)dereliction.
  But even with the method, we do not know how to define the commutation.
  From a syntactical perspective, we are not able to properly understand indexed codereliction.
  The indexed dereliction represents a subtyping rule, so we hope that a commutation with this rule can be defined, from this point of view. But this is much harder for the indexed codereliction.
  \par However, this indexed codereliction rule is clear from a semantical point of view, as we will explain in Section~\ref{sec:idill}.
  One would then imagine that, thanks to the semantics, it is possible to deduce how to define a commutation.
  But, as we will explain in Section~\ref{sec:prom}, we do not know how to define an interpretation for the promotion rule in our model.
  These considerations led us to a choice in this work, in order to have a cut elimination procedure.
  We could either study a system with a promotion, or a system with indexed (co)dereliction. Since our aim here is to use this system for taking into account differential equations and their solutions, we have chosen the second option. 
  The first one is studied from a categorical point of view by Pacaud-Lemay and Vienney~\cite{lemay2023graded}.

  This cut-elimination procedure mimics that of $\DiLL$, with a few critical differences. We should properly prove the termination and confluence, neither of which is trivial. For the confluence, the key point is the multiplicative splitting, which can't really be deterministic (at least it is not in natural examples), while it is not clear that a canonical choice will lead to confluence. For the termination, however, we are pretty sure that it works, for the same reason $\DiLL$ cut elimination works :
  \begin{itemize}
  \item The $\bw/ \p$ elimination may have a supplementary case, but it is a case that erases everything, thus it will just speed up the termination.
  \item The $\bc/ \p$ elimination seems much larger that the $\DiLL$ version, but it is just that, while that of $\DiLL$ introduce one pair of $\contr+\bc$ rules, we are introducing many in parallel, which is blowing the reduction time but cannot cause a non-termination as the same process will be repeated a few more times.
  \end{itemize}

\subsection{Relational Model of $\DBSLL$}
   
We will embed the relational semantics of multiplicative linear logic into a full model of $\DBSLL$ (varying on $\Sr$), or, equivalently, grade potential variants of the usual relational model of  $\DiLL_{0}$.
These models are also models of graded linear logic, and thus contain graded digging and graded dereliction.
We hope to convey to the reader some intuitions on how we eventually intend to unfold the interaction with digging and dereliction, but also why such a general interaction scheme requires further refinements.

Following Breuvart and Pagani~\cite{splitting}, we are considering non-free exponential to be able to internalise grading. Those exponential $\oc^\Sr$ are each characterised by a semiring $\Sr$ which may or may not be the same as the gradation semiring. For simplicity, we only consider the case where both are identified.

Due to this restriction, we require the semiring of the gradation to have additional structure. 
A resource semiring $\Sr$ is given by:
\begin{itemize}
\item a semiring $(\Sr,0,+,1,\times)$ with $1$ as the unit of the new associative operation $\times$,
\item that is discrete, i.e., $x+y=1$ implies $x=0$ or $y=0$,
\item that is positive, i.e., $x+y=0$ implies $x=0$ or $y=0$,
\item that is additive splitting, as described in Definition~\ref{def:addsplit},
\item that is multiplicative splitting, as described in Definition~\ref{mult_split}
\end{itemize}

If $\Sr$ is a resource semiring, then the following define a model of linear logic~\cite{splitting,splitting0}:
\begin{itemize}
\item Let $\Rel$ be the category of sets and relations,
\item it is symmetric monoidal with $A\otimes B :=A\times B$\\ and $r\otimes r':=\{(a,a'),(b,b'))\mid (a,b)\in r, (b,b')\in r'\}$
\item it is star-autonomous, and even compact close, with $A^\bot:=A$ and $r^\bot:=\{(a,b)\mid (b,a)\in r\}$,
\item it accepts several exponentials, among which the free one, with multiset, which is the most commonly used, but many more exist. One way to create such an exponential is to look at the free modules over a fixed resource semiring $\Sr$, written $\oc^\Sr:\Rel\rightarrow \Rel$ and defined by
  \begin{itemize}
  \item $\oc^\Sr A:=[A\Rightarrow^f \Sr]$ is the set of functions $f:A\rightarrow \Sr$ finitely supported, i.e., such that $A-f^{-1}(0)$ is finite,
  \item for $r\in\Rel(A,B)$, $\oc^\Sr r$ is defined as the $\Sr$-couplings:
    \[\oc^\Sr r:= \left\{\left(\Bigl(a\mapsto\sum_{b}\sigma(a,b)\Bigr),\Bigl(b\mapsto\sum_{a}\sigma(a,b)\Bigr)\right)\ \middle|\ \sigma:[r\Rightarrow^f \Sr]\right\}\]
  \item the weakening and the contraction are defined standardly using the additive structure of the semiring:
    \[
        \w_A := \left\{\Bigl([],*\Bigr)\right\}
        \quad\quad
        \contr_{A}:=\left\{\Bigl(\bigl(a\mapsto f(a)+g(a)\bigr),(f,g)\Bigr)\ \middle|\ f,g\in\: \oc^\Sr A\right\}  
    \]
    where $[]:=(a\mapsto 0)$
  \item the dereliction and the digging use the multiplicative structure of the semiring:
    \[
        \der_A:=\Bigl\{(\delta_a,a)\Bigl\}
        \quad\quad
        \p_A:=\left\{\left(\Bigl(a\mapsto \sum_fF(f).f(a)\Bigr),F\right)\ \middle|\  F\in\: \oc^\Sr\oc^\Sr A\right\}
    \]
    where $\delta_a(a)=1$ and $\delta_a(b)=0$ for $b\neq a$.
  \item remains the monoidality:
    \[
        \m_\bot:=\{(*,f)\mid f\in\: \oc^\Sr 1\}
        \quad
        \m_{A,B}:=\left\{\left(\left(\sum_{b}\sigma(\_,b),\ \sum_{a}\sigma(a,\_)\right),\sigma\right)\ \middle|\ \sigma\in\:\oc^\Sr(A{\times} B)\right\}  
    \]
  \end{itemize}
\end{itemize}

Naturality and LL-diagrams can be found in~\cite{splitting} and ~\cite{splitting0}; they are actively using positivity, discreteness, additive splitting and multiplicative splitting. The second of those articles shows that the model can be turned into a model for $\BSLL$ with promotion (but without differential) using, as graded exponential, the restriction of $!^\Sr A$ to generalized multisets of correct weight :
$$ !_xA := \left\{ f\in!^\Sr A \ \middle|\ x\ge\sum_{a\in A}f(a) \right\}. $$
Everything else (functoriality and natural transformation) is just the restriction of the one above to the correct strata.

In a nutchel, Breuvart and Pagani defined models of $\BSLL$ (with promotion but no diferential) by first building models of $\LL$ which are ``well stratified along $\Sr$'', and then grading them by separating strata, which works on-the-nose. In order to build models of $\DBSLL$, we do follow the same pattern but starting from models of promotion-free differential linear logic that are ``well stratified along $\Sr$'', which strata will further unfold into a model of $\DBSLL$.

\begin{thm}
  For any given resource semiring $\Sr$, the above model of linear logic is a model of promotion-free differential linear logic, implementing the codereliction, coweakening and cocontraction as the inverse relation:
  $$ \bd_A:= \{(a,b)\mid (b,a)\in \der_A\}\quad \bw_A:= \{(a,b)\mid (b,a)\in \w_A\}\quad \bc_A:= \{(a,b)\mid (b,a)\in \contr_A\}$$
\end{thm}
\begin{proof}
  Naturality :
  \begin{align*}
    \bd;!r
    &= \{(a,\sum_{a'}\sigma(a',\_))\mid \delta_a=\sum_b\sigma(\_,b), \mathtt{supp}(\sigma)\subseteq r \}\\
    &= \{(a,\sum_{a'}\sigma(a',\_))\mid \exists b, \sigma=\delta_{(a,b)}, \mathtt{supp}(\sigma)\subseteq r\} &\text{(discreteness)}\\
    &= \{(a,\delta_b)\mid (a,b)\in r\}\\
    &= r;\bd\\
    \bw;!r
    &= \{(*,\sum_{a'}\sigma(a',\_))\mid []=\sum_b\sigma(\_,b), \mathtt{supp}(\sigma)\subseteq r \}\\
    &= \{(*,[])\mid \} &\text{(positivity)}\\
    &=\bw\\
    \bc;!r
    &= \{((f,g),\sum_{a'}\sigma(a',\_))\mid f+g =\sum_b\sigma(\_,b), \mathtt{supp}(\sigma)\subseteq r \}\\
    &= \{((f,g),\sum_{a'}\sigma_2(a',\_))\mid f=\sum_b\sigma_1(\_,b),
    \\&\hspace{10em}  g =\sum_b\sigma_2(\_,b), \mathtt{supp}(\sigma_1+\sigma_2)\subseteq r \} &\text{(additive splitting)}\\
    &= \{((\sum_b\sigma_1(\_,b),\sum_b\sigma_2(\_,b)),f+g) \mid \mathtt{supp}(\sigma_1),\mathtt{supp}(\sigma_2)\in r\} &\text{(positivity)}\\
    &= (!r\otimes!r);\bc
  \end{align*}
  Costructural diagrams from~\cite[Sec 2.6]{ehrhardpn}:
  \begin{align*}
    \w;\bw
    &= \{([],[])\}= !\emptyset\\
    \contr;(!r_1\otimes!r_2);\bc
    &=\{((\sum_b\sigma_1(\_,b))+(\sum_b\sigma_2(\_,b)),((\sum_a\sigma_1(a,\_))+((\sum_a\sigma_2(a,\_)))))
    \\&\hspace{10em} \mid \mathtt{supp}(\sigma_1)\subseteq r_1; \mathtt{supp}(\sigma_2)\subseteq r_2\}\\
    &=\{(\sum_b(\sigma_1+\sigma_2)(\_,b),\sum_a(\sigma_1+\sigma_2)(a,\_))\mid \mathtt{supp}(\sigma_1)\subseteq r_1; \mathtt{supp}(\sigma_2)\subseteq r_2\}\\
    &= \{(\sum_b \sigma(\_,b),\sum_a\sigma(a,\_)) \mid \mathtt{supp}(\sigma)\subseteq r_1\cup r_2\}\\
    &=!(r_1\cup r_2) 
  \end{align*}
  \begin{align*}
    (\bw\otimes \id);\m &=\{\}\\
    &=\{((*,\sum_a\sigma(a,\_)),\sigma) \mid []=\sum_b\sigma(\_,b)\}\\
    &=\{((*,[]),[])\} &\text{(positivity)}\\
    &=(\id\otimes \w);\lambda;\bw\\
    (\bd\otimes\id);\m
    &= \{((a,\sum_{a'}\sigma(a',\_)),\sigma) \mid \delta_a=\sum_b\sigma(\_,b)\} \\
    &= \{((a,\delta_a),\delta_{(a,b)})\} &\text{(discreteness)}\\
    &= (\id\otimes \der);\bd\\
    (\bc\otimes\id);\m
    &= \{(((f,g),\sum_a\sigma(a,\_)),\sigma) \mid f+g=\sum_b\sigma(\_,b)\}\\
    &= \{(((\sum_b\sigma_1(\_,b),\sum_b\sigma_2(\_,b)),\sum_a(\sigma_1+\sigma_2)(a,\_)),\sigma_1+\sigma_2) \} &\text{(add. split.)}\\
    &=\{(((f,g),h+k),\sigma_1+\sigma_2) \mid f=\sum_b\sigma_1(\_,b), 
    \\&\hspace{5em} h=\sum_a\sigma_1(a,\_),\ g=\sum_b\sigma_2(\_,b),\ k=\sum_a\sigma_2(a,\_)\}\\
    &= (\id\otimes \contr);\texttt{iso};(\m\otimes \m);\bc \tag*{\qed}
  \end{align*} 
\renewcommand{\qed}{}
\end{proof}

\begin{thm}
This  model can be turned into a model for $\DBSLL$ using, as graded exponential, the restriction of $\oc^\Sr A$ to generalized multisets of correct weight :
$$ \oc_xA := \left\{ f\in\oc^\Sr A \ \middle|\ x\ge\sum_{a\in A}f(a) \right\}. $$
Everything else (functoriality and natural transformation) is just the restriction of the one above to the correct stratum.
\end{thm}

Since those models interpret both the full linear logic and the promotion-free linear logic, we could naively conjecture that they are models of the full differential linear logic, and could eventually be graded into a model of $\DBSLL$ with a proper graded promotion.

Pacaud-Lemay and Vienney~\cite{lemay2023graded} have defined an extension of $\DBSLL$ together with a similar relational interpretation that is basically the one below. However, they did not dwell on the ``resource semiring'' constraints and gave the semantics for $\Sr=\mathbb N$, which corresponds to the gradation of the usual multiset exponential seen as a model of $\DiLL_0$. If one wants to generalise without the resource constraints, it is at the price of the non-commutation of some diagram (functoriality, some naturality, and/or preservation of the semantics through cut elimination).

In their full generality, however, these models do not fully support the promotion in the sense that they do not respect the interactions between the promotion and the costructural morphisms. According to~\cite[Sec 2.6]{ehrhardpn}, we need three other equations to be verified:
\[\bw;\p=\m_\bot;!\bw\quad,\quad\quad\quad\quad \bd;\p=\lambda;((\bw;\p)\otimes(\bd;\bd));\bc\quad\quad \text{and}\quad \quad \bc;\p=(\p\otimes \p);\m;!\bc\ .\]
The first needs the semiring to be an integral domain, which is somehow expected, the second needs a more unusual property, that $xy=1$ implies $x=y=1$. For the third one, the required properties needed on the semiring are still an open question.

\section{An indexed differential linear logic}\label{sec:idill}
In the previous section, we have defined a logic $\DBSLL$ as the syntactical differential of an indexed linear logic $\BSLL$, with its cut elimination procedure.
It is a syntactical differentiation of $\BSLL$, as it uses the idea that differentiation is expressed through costructural rules that mirror the structural rules of $\LL$.
Here we will take a semantical point of view: starting from differential linear logic, we will index it with LPDOcc into a logic named $\IDiLL$, and then study the relation between $\DBSLL$ and $\IDiLL$.

\subsection{IDiLL: a generalization of D-DiLL}\label{ssec:gen-ddill}
As we saw in Section~\ref{sec:LL}, Kerjean gave an alternative version of $\bd$ and $\der$ in previous work~\cite{ddill}, with the idea that in $\DiLL$, the codereliction corresponds to the application of the differential operator $D_0$ whereas the dereliction corresponds to the resolution of the differential equation associated to $D_0$, with a linear map as parameter.

This led to a logic $\DDiLL$, where $\bd$ and $\der$ have the same effect but with a LPDOcc $D$ instead of $D_0$, and where the exponential connectives are indexed by this operator $D$.
One would expect that this work could be connected to $\DBSLL$, but these definitions clash with the traditional intuitions of graded logics.
The reason is syntactical: in graded logics, the exponential connectives are indexed by elements of an algebraic structure, whereas in $\DDiLL$ only one operator is used as an index.
We then change the logic $\DDiLL$ into a logic $\IDiLL$, which is much closer to what is done in the graded setting.
This logic will have the following grammar:
\[
  A,B := 0 \mid 1 \mid \top \mid \bot \mid A \otimes B \mid A \parr B \mid A \with B \mid A \oplus B
\] 
\[
  E,F := \wn_D A \mid \oc_D A \mid E \otimes F \mid E \parr F \mid E \with F \mid E \with F \mid E \oplus F
\]
which is the same grammar as $\DDiLL$, but the exponentials are non-polarised.
The reader should note that here, exponential connectives can appear only once in a formula, which is not the case in $\DBSLL$.
We need this restriction in $\IDiLL$ for semantical reasons. 
As we explained in Section~\ref{subsec:distrib}, spaces of functions and distributions must be defined on finite-dimensional vector spaces, which enforces this restriction.
A more precise discussion on this question is done in Section~\ref{ssec:grad_prom}.
\par In this new framework, we will consider the composition of two LPDOcc as our monoidal operation. Indeed, thanks to Proposition \ref{prop:Ddistrconv}, we have that $D_1(\phi) \ast D_2(\psi) = (D_1 \circ D_2)(\phi \ast \psi)$. The convolution $\ast$ being the interpretation of the cocontraction rule $\bc$, the composition is the monoidal operation on the set of LPDOcc that we are looking for. Moreover, the composition of LPDOcc is commutative, which is a mandatory property for the monoidal operation in a graded framework.
We describe the exponential rules of $\IDiLL$ in Figure~\ref{fig:idillexp}.
\begin{figure}
    \begin{center}
        \AXC{$\vdash \Gamma$}
        \RL{$\w_I$}
        \UIC{$\vdash \Gamma, \wnD A$}
        \DisplayProof
        \qquad\qquad
        \AXC{$\vdash \Gamma,\wn_{D_1} A,\wn_{D_2} A$}
        \RL{$\contr$}
        \UIC{$\vdash \Gamma,\wn_{D_1 \circ D_2} A$}
        \DisplayProof
        \qquad\qquad
        \AXC{$\vdash \Gamma,\wn_{D_1} A$}
        \RL{$\der_I$}
        \UIC{$\vdash \Gamma,\wn_{D_1 \circ D_2} A$}
        \DisplayProof
    \end{center}
    \begin{center}
        \AXC{}
        \RL{$\bwi$}
        \UIC{$\vdash \ocD A$}
        \DisplayProof
        \qquad\qquad
        \AXC{$\vdash \Gamma, \oc_{D_1} A$}
        \AXC{$\vdash \Delta, \oc_{D_2} A$}
        \RL{$\bc$}
        \BIC{$\vdash \Gamma, \Delta, \oc_{D_1 \circ D_2} A$}
        \DisplayProof
        \qquad\qquad
        \AXC{$\vdash \Gamma,\oc_{D_1} A$}
        \RL{$\bdi$}
        \UIC{$\vdash \Gamma,\oc_{D_1 \circ D_2} A$}
        \DisplayProof
    \end{center}
    \caption{Exponential rules of $\IDiLL$}\label{fig:idillexp}
\end{figure}

The indexed rules $\der_D$ and $\bd_D$ of $\DDiLL$ are generalised to rules $\der_I$ and $\bd_I$ involving a variety of LPDOcc, while rules $\der$ and $\bd$ are ignored for now (see the first discussion of Section~\ref{sec:prom}). The interpretations of~$\wn_D A$ and~$\oc_D A$, and hence the typing of $\der_I$ and $\bd_I$ are changed from what $\DDiLL$ would have directly enforced (see remark~\ref{rem:ddillold}).
Our new interpretations for~$\wn_D A$ and~$\oc_D A$ are now compatible with the intuition that in graded logics, rules are supposed to add information.
\begin{align*}
&\sem{?_D A} := \{ g \mid \exists f \in \sem{\wn A},\ D(g)=f \} \qquad && \sem{\oc_D A}  := (\sem{?_D A}')' = \hat{D}(\sem{\oc A})  \\
&\der_I : \sem{\wn_{D_1} A} \to \sem{\wn_{D_1 \circ D_2} A} \qquad && \bd_{I}  : \sem{\oc_{D_1} A} \to \sem{\oc_{D_1 \circ D_2} A} 
\end{align*}

The reader might note that these new definitions have another benefit: they ensure that the dereliction (resp. the codereliction) is well typed when it consists in solving (resp. applying) a differential equation.
This will be detailed in Section~\ref{ssec:semantics}.

Notice that a direct consequence of Proposition \ref{prop:Ddistrconv} is that for two LPDOcc $D_1$ and $D_2$, $\Phi_{D_1 \circ D_2} = \Phi_{D_1} \ast \Phi_{D_2}$. It expresses that our monoidal law is also well-defined with respect to the interpretation of the indexed dereliction. 

\begin{rem}
    \label{rem:ddillold}
    Our definition for indexed connectives and thus for the types of~$\der_D$ and~$\bd_D$ is not only a generalisation but also a dualization of the original one in $\DDiLL$ \cite{ddill}. Kerjean gave types $\der_D : \wn_{D,old} E' \to \wn E'$ and~$\bd_D : \oc_{D,old} E \to  \oc E $. However, graded linear logic carries different intuitions: indices are here to keep track of the operations made through the inference rules. As such, $\der_D$ and $\bd_D$ should introduce indices $D$ and not delete them. Compared with work in~\cite{ddill}, we then change the interpretation of $\wn_D A$ and $\oc_D A$, and the types of $\der_D$ and $\bd_D$.
    \end{rem}

\subsection{Grading linear logic with differential operators}
In this section, we will show that $\IDiLL$ and $\DBSLL$ are co-interpretable.
In order to connect $\IDiLL$ with our results from Section \ref{sec:dbsll}, we have to study the algebraic structure of the set of linear partial differential operators with constant coefficients~$\D$. 
More precisely, our goal is to prove the following theorem.

\begin{thm}\label{thm:addsplit}
    The set $\D$ of LPDOcc is an additive splitting monoid under composition, with the identity operator $id$ as the identity element. 
\end{thm}
To prove this result, we will use multivariate polynomials: $\R[\Xo] := \bigcup_{n \in \N} \R[X_1,\dots,\allowbreak X_n].$
It is well known that $(\R[\Xo],+,\times,0,1)$ is a commutative ring. Its monoidal restriction is isomorphic to $(\D,\circ,id)$, the LPDOcc endowed with composition, through the following monoidal isomorphism
\begin{center}
    $\fonction{\chi}{\hspace{1cm}(\D, \circ)}{(\R[\Xo], \times)}{\sum_{\alpha \in \N^n} a_\alpha \frac{\partial^{|\alpha|} (\_)}{\partial x^{\alpha}}}{\sum_{\alpha \in \N^n} a_\alpha X^{\alpha_1} \dots X_n^{\alpha_n}}$
\end{center}

The following proposition is crucial in the indexation of $\DBSLL$ by differential operators, since the monoid in $\DBSLL$ has to be additive splitting.
\begin{prop}\label{prop:multsplit}
    The monoid $(\R[\Xo], \times, 1)$ is additive splitting.   
\end{prop}
The proof requires some algebraic definitions to make it more readable.
\begin{defi}
    Let $\Rc$ be a non-zero commutative ring. 
    \begin{enumerate}
    \item An element $u \in \Rc$ is a \textit{unit} if there is $v \in \Rc$ such that $uv = 1$.
    \item An element $x \in \Rc \backslash \{0\}$ is \textit{irreducible} if it is not a unit, and not a product of two non-unit elements.
    \item Two elements $x,y \in \Rc$ are \textit{associates} if $x$ divides $y$ and $y$ divides $x$.
    \item $\Rc$ is a \textit{factorial ring} if it is an integral domain such that for each $x \in \Rc \backslash \{0\}$ there is a unit $u \in \Rc$ and $p_1,\dots,p_n \in \Rc$ irreducible elements such that $x = u p_1 \dots p_n$ 
    and for every other decomposition $vq_1\dots q_m = u p_1 \dots p_n$ (with $v$ unit and $q_i$ irreducible for each $i$) we have $n = m$ and a bijection $\sigma : \{ 1,\dots,n\} \rightarrow \{ 1,\dots , n \}$ such that $p_i$ and $q_{\sigma(i)}$ are associated for each $i$.
    \end{enumerate}
\end{defi}

\begin{proof}[Proof of Proposition~\ref{prop:multsplit}]
For each integer $n$, the ring $\R[X_1,\dots,X_n]$ is factorial.
This classical proposition is, for example, proved in~\cite[2.7 Satz 7]{bosch}.

    Let us take four polynomials $P_1,P_2, P_3$ and $P_4$ in $\R[\Xo]$ such that $P_1 \times P_2 = P_3 \times P_4$. 
    There is $n \in \N$ such that $P_1,P_2,P_3,P_4 \in \R[X_1,\dots,X_n]$.
    \par If $P_1 = 0$ or $P_2 = 0$, then $P_3 = 0$ or $P_4 = 0$, since $\R[X_1,\dots,X_n]$ has integral domain. If for example $P_1 = 0$ and $P_3 = 0$, one can define
    \[
        P_{1,3} = 0 \qquad P_{1,4} = P_4 \qquad P_{2,3} = P_2 \qquad P_{2,4} = 1
    \]
    which gives a correct decomposition. And we can reason symmetrically for the other cases.
    \par Now, we suppose that each polynomial $P_1,P_2,P_3$ and $P_4$ are non-zero. By factoriality of $\R[X_1,\dots,X_n]$, we have a decomposition 
    \begin{equation*}
        P_i = u_i Q_{n_{i-1}+1} \times \dots Q_{n_i} \tag{for each $1 \leq i \leq 4$}
    \end{equation*}
    where $n_0 = 0 \leq n_1 \dots \leq n_4$, $u_i$ are units and $Q_i$ are irreducible. Then, the equality~${P_1P_2 = P_3P_4}$ gives
    \[
        u_1u_2Q_1 \dots Q_{n_2} = u_3u_4 Q_{n_2 + 1} \dots Q_{n_4}.
    \]
    Since $u_1u_2$ and $u_3u_4$ are units, the factoriality implies that $n_2 = n_4 - n_2$ and that there is a bijection $\sigma : \{1,\dots,n_2\} \to \{n_2+1,\dots,n_4\}$ such that $Q_i$ and $Q_{\sigma(i)}$ are associates for each~${1 \leq i \leq n_2}$. 
    It means that for each $1 \leq i \leq n_2$, there is a unit $v_i$ such that $Q_{\sigma(i)} = v_i Q_i$.
    Hence, defining two sets $A_3 = \sigma^{-1}(\{n_2+1,\dots,n_3\})$ and $A_4 = \sigma^{-1}(\{n_3+1,\dots,n_4\})$ we can rewrite our polynomials $P_1$ and $P_2$ using: 
    \begin{align*}
        &A_{1,3} = A_3 \cap \{1,\dots,n_1\} = p_1,\dots,p_{m_1} \quad &R_{1,3} = Q_{p_1} \dots Q_{p_{m_1}} \quad v_{1,3} = v_{p_1} \dots v_{p_{m_1}}& \\
        &A_{1,4} = A_4 \cap \{1,\dots,n_1\} = q_1,\dots,q_{m_2} \quad &R_{1,4} = Q_{q_1} \dots Q_{q_{m_2}} \quad v_{1,4} = v_{q_1} \dots v_{q_{m_2}}& \\
        &A_{2,3} = A_3 \cap \{n_1 + 1,\dots,n_2\} = r_1,\dots,r_{m_3} \quad &R_{2,3} = Q_{r_1} \dots Q_{r_{m_3}} \quad v_{2,3} = v_{r_1} \dots v_{r_{m_3}}& \\
        &A_{2,4} = A_4 \cap \{n_1 + 1,\dots,n_2\} = s_1,\dots,s_{m_4} \quad &R_{2,4} = Q_{s_1} \dots Q_{s_{m_4}} \quad v_{2,4} = v_{s_1} \dots v_{s_{m_4}}&
    \end{align*}
    which leads to 
    \[ P_1 = u_1 R_{1,3}R_{1,4} \hspace{1.5em} P_2 = u_2 R_{2,3}R_{2,4} \hspace{1.5em} P_3 = u_3 v_{1,3}R_{1,3}v_{2,3}R_{2,3} \hspace{1.5em} P_4 = u_4 v_{1,4}R_{1,4}v_{2,4}R_{2,4}
    \]
    Finally, we define our new polynomials
    \[ P_{1,3} = u_1 R_{1,3} \qquad P_{1,4} = R_{1,4} \qquad P_{2,3} = \frac{u_3 v_{1,3} v_{2,3}}{u_1} R_{2,3} \qquad P_{2,4} = \frac{u_1 u_2}{u_3 v_{1,3} v_{2,3}} R_{2,4}
    \]
    gives the wanted decomposition: this is straightforward for $P_1, P_2$ and $P_3$ (the coefficients are chosen for that), and for $P_4$, it comes from the fact that $u_1 u_2 = u_3 u_4$ (which is in the definition of a factorial ring), and that $v_{1,a}v_{1,b}v_{2,a}v_{2,b} = 1$ which is easy to see using our new polynomials $R_{1,3},R_{1,4},R_{2,3},R_{2,4}$ and the equality $P_1P_2 = P_3P_4$.
\end{proof}

This result ensures that $(\D,\circ,id)$ is an additive splitting monoid. 
Then, $\D$ induces a logic \textsf{DB$_\D$LL}. 
In this logic, since the preorder of the monoid is defined through the composition rule, for $D_1$ and $D_2$ in $\D$ we have
\begin{center}
    $D_1 \leq D_2 \Longleftrightarrow \exists D_3 \in \D,\ D_2 = D_1 \circ D_3$
\end{center}
which expresses that the rules $\der_I$ and $\bdi$ from $\IDiLL$ and those from \textsf{DB$_\D$LL} are exactly the same. 
In addition, the weakening and the coweakening from \textsf{DB$_\D$LL} are rules which exist in $\IDiLL$ (the (co)weakening with $D = id$), and a weakening (resp. a coweakening) in $\IDiLL$ can be expressed in \textsf{DB$_\D$LL} as an indexed weakening (resp. an indexed coweakening). In fact, this indexed weakening is the one that appears in the cut elimination procedure of $\DBSLL$. Hence, this gives the following proposition.

\begin{prop}
    Each rule of $\IDiLL$ is admissible in \textsf{DB$_\D$LL}, and each rule of \textsf{DB$_\D$LL} except $\der$ and $\bd$ is admissible in $\IDiLL$.
\end{prop}

With this proposition, Theorem~\ref{thm:cutelim} ensures that $\IDiLL$ enjoys a cut elimination procedure, which is the same as the one defined for $\DBSLL$.
This procedure is even easier to define in the case of $\IDiLL$.
One issue in the definition of the cut elimination of $\DBSLL$ is the fact that indexed (co)derelictions cannot go up in the tree when they act on a formula introduced by a (co)weakening. 
To deal with this issue, we have introduced these rules $\w_I$ and $\bwi$, which are admissible in $\DBSLL$.
This gives an additional step to prove the cut elimination theorem of $\DBSLL$.
However, if one wants to prove directly the cut elimination theorem in $\IDiLL$, it will then be easier, since $\w_I$ and $\bw_I$ will not need to be defined during the proof, because they already exist in the logic.
This additional step will hence not exist in the proof.
\par One could then define these rules directly in $\DBSLL$, but we have chosen to keep this logic as close to graded linear logic as possible, and then to have only non-indexed (co)weakenings as primitive rules in $\DBSLL$.
\par The cut elimination procedure for $\IDiLL$ is then exactly the one of $\DBSLL$.

\subsection{A concrete semantics for IDiLL}\label{ssec:semantics}
Now that we have defined the rules and the cut elimination procedure for a logic able to deal with the interaction between differential operators in its syntax, we should express how it semantically acts on smooth maps and distributions.
For $\MALL$ formulas and rules, the interpretation is the same as the one for $\DiLL$ (or $\DDiLL$), given in Section \ref{sec:LL}.
First, we give the interpretation of our indexed exponential connectives. 
Beware that we are still here in a \emph{finitary setting}, in which exponential connectives only apply to finite-dimensional vector spaces, meaning that $\br{A}=\R^n$ for some $n$ in equation (\ref{eq:semindexedexp}) below. 
This makes sense syntactically as long as we do not introduce a promotion rule, and corresponds to the denotational model exposed originally by Kerjean. 
As mentioned in the conclusion, we think that work in higher-dimensional analysis should provide a higher-order interpretation for indexed exponential connectives~\cite{Gannoun_dualite}.

Let us take $D \in \D$. Considering the formal sum associated to the operator $D$, this sum can be applied to any $f \in \Ci(\R^n,\R)$ for any $n$, regardless of the order of $D$, by injecting smoothly $\Ci(\R^n,\R) \subseteq \Ci(\R^m,\R)$ for any $m \geq n$. We give the following interpretation of graded exponential connectives:
\begin{align}
\label{eq:semindexedexp}
    \br{\oc_D A} &:= \hat{D}(\sem{\oc A}) \nonumber \\
    \br{?_D A} &:= \{ f \in \Ci(\br{A}',\R) \mid \exists g \in \Ci(\br{A}',\R),\ D(f) = g\} = D^{-1}(\sem{\wn A})
\end{align}
We recall that $\hat{D}$ appears in the definition of the application of a LPDOcc to a distribution, see equation \ref{eq:lpdodistrib}.

From this definition, one can note that when $D = id$, we get 
\[ 
    \br{\oc_{id} A} = (\Ci(\br{A},\R))' = \br{! A} \qquad \qquad \br{?_{id} A} = \Ci(\br{A}',\R) = \br{? A}.
\]

\begin{rem}
One can notice that, as differential equations always have solutions in our case, the space of solutions $\sem{\wn_D A}$ is \emph{isomorphic} to the function space $\sem{\wn A}$. The isomorphism in question is plainly the dereliction $\der_D : f \mapsto \Phi_D \ast f$. While our setting might be seen as too simple from the point of view of analysis, it is a first and necessary step before extending $\IDiLL$ to more intricate differential equations, for which these spaces would not be isomorphic since $\Phi_D$ would not exist. If we were to explore the abstract categorical setting for our model, these isomorphisms would be relevant in a \emph{bicategorical} setting, with LPDO as 1-cells. Hence, the 2-cells would be isomorphisms if one restricts to LPDOcc, but much more complicated morphisms may appear in the general case. 
\end{rem}

The exponential modality $\oc_D$ has been defined on finite-dimensional vector spaces. It can be extended into a functor, i.e. as an operation on maps acting on finite-dimensional vector spaces.
The definition is the following: for $f : E \multimap F$ a linear map between two vector spaces $E$ and $F$, we define
\[
    \fonction{\oc_D f}{\oc_D E}{\oc_D F}{\psi \circ D}{(g \in \Ci(F',\R) \mapsto \psi \circ D (g \circ f)).}
\]

The next step is to give a semantical interpretation of the exponential rules. Most of these interpretations will be quite natural, in the sense that they will be based on the intuitions given in Section \ref{ssec:gen-ddill} and on the model of $\DiLL$ described in previous work~\cite{ddill}. 
However, the contraction rule will require some refinements.
The contraction takes two formulas $\wn_{D_1} A$ and $\wn_{D_2} A$, and contracts them into a formula $\wn_{D_1 \circ D_2} A$.
In our model, it corresponds to the contraction of two functions $f \in \Ci(E',\R)$ such that $D_1(f) \in \Ci(E',\R)$ and $g \in \Ci(E',\R)$ such that $D_2(g) \in \Ci(E',\R)$ into a function $h \in \Ci(E',\R)$ such that~$D_1 \circ D_2(h) \in \Ci(E',\R)$.
In differential linear logic, the contraction is interpreted as the pointwise product of functions (see section~\ref{sec:LL}). 
This is not possible here, since we do not know how to compute $D_1 \circ D_2 (f \cdot g)$.
We will then use the fundamental solution, which has the property that $D(\Phi_D \ast f) = f$.
This leads to the following definition.
\begin{defi}\label{def:rulesem}
    We define the interpretation of each exponential rule of $\IDiLL$ by:
    \begin{align*}
        &\fonction{\w}{\R}{\wn_{id} E}{1}{cst_1}
        &&\fonction{\bw}{\R}{\oc_{id} E}{1}{\delta_0} \\
        &\fonction{\contr}{\wn_{D_1} E \ \hat{\otimes} \ \wn_{D_2} E}{\wn_{D_1 \circ D_2} E}{f \otimes g}{\Phi_{D_1 \circ D_2}\ast (D_1(f)\cdot D_2(g))}
        &&\fonction{\bc}{\oc_{D_1} E \ \hat{\otimes} \ \oc_{D_2} E}{\oc_{D_1 \circ D_2} E}{\psi \otimes \phi }{\psi \ast \phi} \\
        &\fonction{\der_I}{\wn_{D_1} E}{\wn_{D_1 \circ D_2} E}{f}{\Phi_{D_2} \ast f} 
        &&\fonction{\bdi}{\oc_{D_1} E}{\oc_{D_1 \circ D_2} E}{\psi}{\psi \circ D_2}
    \end{align*}
\end{defi}

\begin{rem}
    One can note that we have only defined the interpretation of the (co)weakening when it is indexed by the identity. This is because, as well as for $\DBSLL$, the one of $\w_I$ and $\bwi$ can be deduced from this one, using the definition of $\der_I$ and $\bdi$. This leads to
    \[ \w_I : 1 \mapsto \Phi_D \ast cst_1 = cst_{\Phi_D(cst_1)} 
    \qquad\qquad
    \bwi : 1 \mapsto \delta_0 \circ D = (f \mapsto D(f)(0)).
    \]
\end{rem}

\paragraph{Polarized multiplicative connectives}
The interpretation for $\bc$ and $\contr$ is justified by the fact that in Nuclear Fréchet or Nuclear DF-spaces~\cite{ddill}, both the $\parr$ and $\otimes$ connectors of $\LL$ are interpreted by the same completed topological tensor product $\hat{\otimes}$. They, however, do not apply to the same kind of spaces, as $\wn E$ is Fréchet while $\oc E$ is not. Thus, basic operations on the interpretation of $A \parr B$ or $A \otimes B$ are first defined on elements $a \otimes b$ on the tensor product, and then extended by linearity and completion. The duality between $A'\parr B'$ and $A \otimes B$ is the one derived from function application and scalar multiplication. A function $\ell^A \otimes \ell^B \in A'\parr B' $ acts on $A \otimes B$ as $ \ell^A \otimes \ell^B  : x \otimes y \in A \otimes B \mapsto \ell^A(x) \cdot \ell^{B}(y) \in \R$.

\begin{rem}\label{rem:diracs}
    In order to define a linear morphism $m$ from $\oc E$, one can define the action of this morphism on each dirac distribution $\delta_x$ for each $x \in E$, which is an element of $\oc E$, and extend it by linearity and completion. From the Hahn-Banach theorem, the space of linear combinations of $\{\delta_x \mid x \in E\}$ is dense in $\Ci(E,\R)$, which justifies this technique. It can be extended to the definition of linear morphisms from $\oc_D E$, just by post-composing with the operator $D$. 
\end{rem}

Now we want to address the fact that the interpretation of $\contr$ is way less elegant than the interpretation of $\bc$. We are sadly aware of this issue, and working on a generalised version of $\IDiLL$ that solves this thanks to the introduction of a Laplace operator~\cite{kerjean_lemay23}. Meanwhile, let us point out that the dual version of $\contr$ is a tad more elegant, and that the ugliness of its interpretation is, in fact, hidden in the dualization of it.

\begin{prop}
\label{prop:dual_contr}
The dual of the contraction law corresponds to \[\contr'_{D_1,D_2} : 
\delta_x \circ (D_1 \circ D_2) \in \oc_{D_1 \circ D_2} E   \mapsto  ((\delta_x \circ D_1) \otimes (\delta_x \circ D_2)) \in 
\oc_{D_1} E \otimes \oc_{D_2} E.\]
\end{prop}

\begin{proof}
Because we are working on finite-dimensional spaces $E$, an application of the Hahn-Banach theorem gives us that the span of $\{\delta_x \mid x \in E \}$ is dense in $\oc E$. 
As such, the interpretation of $\contr'$ can be restricted to elements of the form $\delta_x \circ {D_1 \circ D_2} \in \oc_{D_1 \circ D_2} E$.
Remember also that for a linear map $\ell : E \multimap F$, its dual $\ell ' : F'  \multimap E'$ computes as follows : 
\[ \ell' :  h \in F' \mapsto (x \in E \mapsto h(\ell(x))) \in F \]

Indeed, consider $\ell \in  (\wn_{D_1 \circ D_2} E)'$. As all the spaces considered are reflexive, one has:
\[ \oc_{D_1 \circ D_2} E \simeq \{ \phi \circ D_1 \circ D_2 \mid  \oc E \} \]
and as such there is $\phi \in \oc E$ such that $\ell = \phi \circ D_1 \circ D_2$. As such, for any $f \otimes g \in \wn_{D_1} E  \hat{\otimes}  \wn_{D_2} E$ one has:
\begin{align*}
(\ell \circ \contr) (f \otimes g) & = (\phi \circ D_1 \circ D_2) (\Phi_{D_1 \circ D_2} \ast (D_1(f)\cdot D_2(g))) \\ 
& = \phi (D_1(f) . D_2 (g)) 
\end{align*}
Considering $\phi=\delta_x$, we obtain 
\begin{align*}
(\ell \circ \contr) (f \otimes g) & = \delta_x (D_1(f) . D_2 (g)) \\
& = \delta_x (D_1(f))  \cdot \delta_x (D_2 (g)) \\ 
&=  ((\delta_x \circ D_1) \otimes (\delta_x \circ D_2)) (f \otimes g).
\end{align*}
Hence $\contr'$ corresponds to $\contr'_{D_1,D_2} : \oc_{D_1 \circ D_2} E  \to \oc_{D_1} E \otimes \oc_{D_2} E$.
\end{proof}

\par In order to ensure that Definition \ref{def:rulesem} gives a correct model of $\IDiLL$, we should verify the well-typedness of each morphism.
First, this is obvious for the weakening and the coweakening. The function $cst_1$ defined on $E$ is smooth, and $\delta_0$ is the canonical example of a distribution.
Moreover, we interpret $\w$ and $\bw$ in the same way as in the model of $\DiLL$ on which our intuitions are based.
The indexed dereliction is well-typed, because for $f \in \wnu$, there is $g \in \Ci(E',\R)$ such that $D_1(f) = g$ by definition.
Hence, ${D_1 \circ D_2(\Phi_{D_2}\ast f) = D_1(f) = g \in \Ci(E',\R)}$ so $\der_I(f) \in \wnud$.
For the contraction, if~${f \in \wnu}$ and~$g \in \wnd$, $D_1(f)$ and $D_2(g)$ are in $\Ci(E',\R)$, and so is their scalar product.
Hence, ${D_1 \circ D_2(\contr(f\otimes g)) = D_1(f)\cdot D_2(g) }$ which is in $\Ci(E',\R)$.
The indexed codereliction is also well-typed: for $\psi \in \ocu$, equation (\ref{eq:semindexedexp}) ensures that $\psi = \hat{D_1}(\psi_1)$ with $\psi_1 \in \oc E$, so~$\psi \circ D_2 = (\psi_1 \circ D_1) \circ D_2 \in \ocud$.
Finally, using similar arguments for the cocontraction, if $\psi \in \ocu$ and $\phi \in \ocd$, then $\psi = \hat{D_1}(\psi_1)$ and $\phi = \hat{D_2}(\phi_1)$, with $\psi_1, \phi_1 \in \oc E$. Hence, 
\[ \psi \ast \phi = (\psi_1 \circ D_1) \ast (\phi_1 \circ D_2) = (\psi_1 \ast \phi_1) \circ (D_1 \circ D_2) = \widehat{D_1 \circ D_2}(\psi_1 \ast \phi_1) \in \ocud.\]
We have then proved the following proposition.
\begin{prop}
    Each morphism $\w,\bw,\contr,\bc,\der_I$ and $\bdi$ is well-typed.
\end{prop}
Another crucial point to study is the compatibility between this model and the cut elimination procedure $\lto$. 
In denotational semantics, one would expect that a model is invariant with respect to the computation.
In our case, that would mean that for each step of rewriting of $\lto$, the interpretation of the proof-tree has the same value.
\par It is easy to see that this is true for the cut $\w_I/\bwi$, since $D(\Phi_D \ast cst_1)(0) = cst_1(0) = 1$.
For the cut between a contraction and an indexed coweakening, the interpretation before the reduction is $\delta_0(D_1 \circ D_2)(\Phi_{D_1 \circ D_2}(D_1(f)\cdot D_2(g))) = D_1(f)(0)\cdot D_2(g)(0)$, which is exactly the interpretation after the reduction.

Finally, proving the invariance of our semantics over the cut between a contraction or a weakening and a cocontraction takes slightly more work. 
The weakening case is enforced by linearity of the distributions, while the contraction case relies on the density of $\{ \delta_x \mid x \in E\}$ in $\oc E$.
 
\begin{lem}
The interpretation of $\DBSLL$ with $\D$ as indices is invariant over the $\contr / \bc$ and the $\bc/\w_I$ cut-elimination rules, as given in Figure \ref{fig:dbsllcut}.
\end{lem}

\begin{proof} Before cut-elimination, the interpretation of the $\bc / \w$ as given in Figure \ref{fig:dbsllcut} is:
    \begin{align*}
        &(\psi \ast \phi)(\Phi_{D_1 \circ D_2} \ast cst_1) \\
        &= \psi(x \mapsto \phi(y\mapsto \Phi_{D_1} \ast (\Phi_{D_2} \ast cst_1)(x+y))) \\
        &= \psi(x \mapsto \phi(y\mapsto \Phi_{D_1}(z \mapsto \Phi_{D_2} \ast cst_1(x+y-z)))) \\
        &= \psi(x \mapsto \phi(y\mapsto \Phi_{D_1}(cst_{\Phi_{D_2}(cst_1)}))) \\
        &= \psi(x \mapsto \phi(y\mapsto \Phi_{D_1}(\Phi_{D_2}(cst_1).cst_1))) \\
        &= \psi(x \mapsto \phi(y\mapsto \Phi_{D_2}(cst_1).\Phi_{D_1}(cst_1))) \tag{by homogeneity of $\phi$} \\
        &= \psi(x \mapsto \phi(cst_{\Phi_{D_2}(cst_1).\Phi_{D_1}(cst_1)})) \\
        &= \psi(x \mapsto \phi(\Phi_{D_1}(cst_1).cst_{\Phi_{D_2}(cst_1)})) \\
        &= \psi(x \mapsto \Phi_{D_1}(cst_1).\phi(cst_{\Phi_{D_2}(cst_1)})) \tag{by homogeneity of $\phi$} \\
        &= \psi(cst_{\Phi_{D_1}(cst_1).\phi(cst_{\Phi_{D_2}(cst_1)})}) \\
        &= \psi(\phi(cst_{\Phi_{D_2}(cst_1)}).cst_{\Phi_{D_1}(cst_1)}) \\
        &= \phi(cst_{\Phi_{D_2}(cst_1)}).\psi(cst_{\Phi_{D_1}(cst_1)}) \tag{by homogeneity of $\psi$}
    \end{align*}
which corresponds to the interpretation of the proof after cut-elimination.

Let us tackle now the $\bc / \contr$ cut-elimination case.
Suppose that we have $D_1, D_2, D_3, D_4 \in \D$ such that $D_1 \circ D_2 = D_3 \circ D_4 $.  By the additive splitting property we have $D_{1,3},D_{1,4},D_{2,3},D_{2,4}$ such that 
\[ D_1 = D_{1,3} \circ D_{1,4} \qquad D_2 = D_{2,3} \circ D_{2,4} \qquad D_3 = D_{1,3} \circ D_{2,3} \qquad D_4 = D_{1,4} \circ D_{2,4}.\]
The diagrammatic translation of the cut-elimination rule in Figure \ref{fig:dbsllcut} is the following. 
\begin{center}
\begin{tikzcd}
\oc_{D_1} E \otimes \oc_{D_2} E
\arrow[d,"\bc_{D_1,D_2}"] \arrow[rrr,"\contr'_{D_{1,3},D_{1,4}}\otimes \contr'_{D_{2,3},D_{2,4}}"]  &  &  & {\oc_{D_{1,3}} E \otimes \oc_{D_{1,4}} E \otimes \oc_{D_{2,3}} E \otimes \oc_{D_{2,4}} E} \arrow[dd]  \\
\oc_{D_1 \circ D_2} E = \oc_{D_3 \circ D_4} E \arrow[d,"\contr'_{D_3,D_4}"] &  &  &                                                                                                 \\
\oc_{D_3} E \otimes \oc_{D_4} E                         &  &  & {\oc_{D_{1,3}} E \otimes \oc_{D_{2,3}} E \otimes \oc_{D_{1,4}} E \otimes \oc_{D_{2,4}} E} \arrow[lll,"\bc_{D_{1,3},D_{2,3}}\otimes \bc_{D_{1,4},D_{2,4}}"]
\end{tikzcd}
\end{center}

Remember that the convolution of Dirac operators is the Dirac of the sum of points, and as such, we have :
\[\bc_{D_a,D_b}  : (\delta_x \circ {D_a}) \otimes (\delta_y \circ D_b) \mapsto (\delta_{x+y} \circ D_b \circ D_a) .\] We make use of proposition~\ref{prop:dual_contr} to 
 compute easily that the diagram above commutes on elements $(\delta_x \circ {D_1}) \otimes (\delta_y \circ D_2)$ of $\oc_{D_1} E \otimes \oc_{D_2} E$, and as such commutes on all elements by density and continuity of $\bc$ and $\contr'$.
\end{proof}

\par In order to ensure that this model is fully compatible with $\lto$, it also has to be invariant by $\lto_{\der_I}$ and by $\lto_{\bdi}$.
For $\lto_{\der_I}$, the interpretation of the reduction step when the indexed dereliction meets a contraction is 
\begin{align*}
    & \Phi_{D_3} \ast (\Phi_{D_1 \circ D_2}\ast (D_1(f)\cdot D_2(g))) \\
    &= \Phi_{D_1 \circ D_2 \circ D_3}\ast ((D_1(f)\cdot D_2(g)).cst_1)\\
    &= \Phi_{D_1 \circ D_2 \circ D_3}\ast ((D_1(f)\cdot D_2(g))\cdot D_3(\Phi_{D_3}\ast cst_1))\\
    &= \Phi_{D_1 \circ D_2 \circ D_3}\ast (D_1 \circ D_2(\Phi_{D_1 \circ D_2}\ast (D_1(f)\cdot D_2(g)))\cdot D_3(\Phi_{D_3}\ast cst_1))
\end{align*} 
which is the interpretation after the application of $\lto_{\der_I,2}$.
The case with a weakening translates the fact that $\Phi_{D_1 \circ D_2} = \Phi_{D_1} \ast \Phi_{D_2}$.
Finally, the axiom rule introduces a distribution $\psi \in \ocu$ and a smooth map $f \in \ocu$, and $\lto_{\der_I,4}$ corresponds to the equality $\Phi_{D_1 \circ D_2} \ast D_1(f) = \Phi_{D_2} \ast f$.
\par The remaining case is the procedure $\lto_{\bdi}$, which is quite similar to $\lto_{\der_I}$. The invariance of the model with the cocontraction case follows from Proposition \ref{prop:Ddistrconv}.
For the weakening, this is just the associativity of the composition, and the axiom works because $\delta_0$ is the neutral element of the convolution product.
We can finally deduce that our model gives an interpretation which is invariant under the cut elimination procedure of Section \ref{sec:dbsll}.
\begin{prop}
    Each morphism $\w,\bw,\contr,\bc,\der_I$ and $\bdi$ is compatible with the cut elimination procedure $\lto$.
\end{prop}

\section{Conclusion}

In this paper, we define a multi-operator version of $\DDiLL$, which turns out to be the finitary differential version of Graded Linear Logic. We describe the cut-elimination procedure and give a denotational model of this calculus in terms of differential operators. This provides a new and unexpected semantics for Graded Linear Logic, and tightens the links between Linear Logic and Functional Analysis.
\subsection{Related work}
This work is an attempt to give notions of differentiation in programming language semantics.
Recently, other works have made some advances in this direction.
We compare our approaches and explain the choices that we have made. 
\paragraph*{Graded differential categories} In recent works, Pacaud-Lemay and Vienney have defined a graded extension of differential categories~\cite{lemay2023graded}. 
If one wants to give a categorical semantics for $\IDiLL$, their work is a natural starting point.
However, some major differences have to be noted.
First, since they follow what has been done in graded logics, their indices are elements of some semiring, whose elements are not necessarily differential operators.
Secondly, they do not have indexed derelictions and coderelictions. 
While we use these rules to solve or apply differential equations, their notion of differentiation comes from the non-indexed codereliction, which is the usual point of view in $\DiLL$, and the indexes are here to possibly refine the notion of differentiability.
More precisely, for a semiring $\Sr$, they define an $\Sr$-graded monoidal coalgebra modality as follows. 
\begin{defi}
  A \emph{$\Sr$-graded monoidal coalgebra modality} on a symmetric monoidal category $(\mathcal{L},\otimes,I)$ is a tuple $(\oc,\p,\der,\contr,\w,\m^\otimes,\m_\bot)$ where:
  \begin{itemize}
    \item for each $s \in \Sr$, $\oc_s : \mathcal{L} \to \mathcal{L}$ is an endomorphism;
    \item for each $s,t \in \Sr$, $\p_{s,t}:\oc_{st}A\to\oc_{s}\oc_{t}A$ and $\contr_{s,t}:\oc_{s+t}A\to\oc_{s}A\otimes\oc_{t}A$ are natural transformations;
    \item $\der : \oc_1 A \to A$ and $\w : \oc_0 A \to I$ are natural transformations;
    \item for each $s\in\Sr$, $\m_s^\otimes : \oc_s A \otimes \oc_s B \to \oc_s(A\otimes B)$ and $\m_{\bot,s}:I\to \oc_s I$ are natural transformations. 
  \end{itemize}
    In addition, some categorical equalities have to be satisfied.
\end{defi}
The equalities are detailed in Definitions 2.1 and 2.2 of~\cite{lemay2023graded}.
This can be extended, with costructural morphisms, in order to capture the notion of differentiation.
\begin{defi}
    A \emph{$\Sr$-graded monoidal additive bialgebra differential modality} on an additive symmetric monoidal category $(\mathcal{L},\otimes,I)$ is a tuple $(\oc,\p,\der,\contr,\w,\m^\otimes,\m_\bot,\bc,\bw,\bd)$ where $(\oc,\p,\der,\contr,\w,\m^\otimes,\m_\bot)$ is a $\Sr$-graded monoidal coalgebra modality on $\mathcal{L}$, and 
    \begin{itemize}
        \item for each $s,t \in \Sr$, $\bc_{s,t} : \oc_s A \otimes \oc_t A \to \oc_{s+t}A$ is a natural transformation;
        \item $\bw : I \to \oc_0 A$ is a natural transformation;
        \item $\bd : A \to \oc_1 A$ is a natural transformation.
    \end{itemize}
    In addition, some categorical equalities have to be satisfied.
\end{defi}
\begin{rem}
    In differential categories, deriving transformations are the natural way to consider differentiation. 
    In their paper, Pacaud-Lemay and Vienney define graded deriving transformations and graded Seely isomorphisms.
    Alternatively, they define graded costructural morphisms ($\bw,\bc$ and $\bd$), and prove that this is equivalent to graded deriving transformation and graded Seely isomorphisms.
    Here, we only consider the second version, with the costructural rules, since it is closer to our work.
\end{rem}
The semantics that we have defined for $\IDiLL$ is not a $\Sr$-graded monoidal additive bialgebra differential modality. 
Of course, the main reason is that the set of LPDOcc is not a semiring, since we do not know which rule would correspond to the product.
This implies that we do not know how to define $\der,\bd$ and $\p$.
However, some natural transformations are still possible to define in our concrete model.
The ones interpreting the logical rules are the ones given in Definition~\ref{def:rulesem}.
But in addition, the transformations $\m^\otimes$ and $\m_\bot$, which express the monoidality of the functors $\oc_D$ can be defined as well.
Using Remark~\ref{rem:diracs}, we define these morphisms on diracs for each LPDOcc $D$ and each finite dimensional vector spaces $E,F$:
\[
    \fonction{\m_D^\otimes}{\oc_D E \otimes \oc_D F}{\oc_D (E \otimes F)}{(\delta_x \circ D) \otimes (\delta_y \circ D)}{\delta_{x \otimes y}\circ D}
    \qquad\qquad
    \fonction{\m_{\bot,D}}{\R}{\oc_D \R}{x}{x \delta_1 \circ D.}
\]
\paragraph{Higher-order models of smooth functions}
Our paper is based on a specific interpretation of finitary $\DiLL$, which was first explained in ~\cite{ddill}. This semantics extends in fact to full Differential Linear Logic,  by describing higher order functions on Fréchet or DF-spaces~\cite{kerjean_lemay19, Gannoun_dualite}. Several other higher-order semantics of $\DiLL$ exist, among them the already mentioned work by Dabrowski~\cite{dabrowski_models} or Ehrhard~\cite{ehrhard_kosequence_2002}. Convenient structures~\cite{kriegl_convenient_1997,diffeology,blute_convenient_2012} also give a model of $\DiLL$ and higher-order differentiation: they share the common idea that a (higher-order) smooth function $f : E \to F$ is defined as a function sending a smooth curve $c : \R \to E$ to a smooth curve $f : \R \to R$. They share particularly nice categorical structure, and enjoy limits, colimits, quotients\dots However, they crucially lack good $\ast$-autonomous structure, on which the present work is build. Specifically, convenient vector spaces do not form a $\ast$-autonomous category, and cannot, due to the use of bornologies~\cite[Section 6]{kerjean_tasson_2016}. Likewise, diffeological spaces enjoy good cartesian structure but do not have any $\ast$-autonomous structure.

\subsection{Promotion and higher-order differential operators}\label{sec:prom}

In Section~\ref{sec:idill}, we have defined a differential extension of graded linear logic, which is interpreted thanks to exponentials indexed by a monoid of differential operators. This extension is done \emph{up-to promotion}, meaning that we do not incorporate promotion in the set of rules. There are two reasons why it makes sense to leave promotion out of the picture:
\begin{itemize}
\item $\DiLL$ was historically introduced without it, with a then perfectly symmetric set of rules.
\item Concerning semantics, LPDOcc are only defined when acting on functions with a finite-dimensional codomain: $D : \Ci (\R^n, \R) \to \Ci (\R^n, \R)$. Introducing a promotion rule would mean extending the theory of LPDOcc to higher-order functions.
\end{itemize}
In this section, we sketch a few of the difficulties one faces when trying to introduce promotion and dereliction rules indexed by differential operators, and explore possible solutions. 

\paragraph{Graded (co)dereliction}
Indexing the promotion goes hand-in-hand with indexing the \emph{dereliction}. In Figure \ref{fig:dbsllexp}, we introduced a basic (not indexed) dereliction and codereliction rule $\der$ and $\bd$. The original intuition of $\DiLL$ is that codereliction computes the differentiation at $0$ of some proof. Following the intuition of $\DDiLL$, dereliction computes a solution to the equation $D_0 (\_) = \ell$ for some $\ell$. Therefore, as indices are here to \emph{keep track} of the computations, and following  equation (\ref{eq:semindexedexp}), we should have (co)derelictions indexed by $D_0$ as below:

\begin{center}
\AXC{$\vdash \Gamma, A$}
\RL{$\bd$}
\UIC{$\vdash \Gamma, \oc A$}
\DisplayProof
\qquad
\AXC{$\vdash \Gamma, A$}
\RL{$\der$}
\UIC{$\vdash \Gamma, \wn A$}
\DisplayProof
\hspace{2cm}
\AXC{$\vdash \Gamma, A$}
\RL{$\bd_{D_0}$}
\UIC{$\vdash \Gamma, \oc_{D_0} A$}
\DisplayProof
\qquad
\AXC{$\vdash \Gamma, A$}
\RL{$\der_{D_0}$}
\UIC{$\vdash \Gamma, \wn_{D_0} A$}
\DisplayProof
\end{center}

Mimicking what happens in graded logics, $D_0$ should be the identity element for the second law in the semiring interpreting the indices of exponentials in $\DBSLL$.
However, $D_0$ is \emph{not} a linear partial differential operator (even less with constant coefficient).
Let us briefly compare how a LPDOcc  $D$ and $D_0$ act on a function $f \in \Ci (\R^n, \R)$:
\[D : f \mapsto \left( y \in \R^n \mapsto \sum_{\alpha \in \N^n} a_\alpha \frac{\partial^{|\alpha|} f}{\partial x^{\alpha}}(y) \right) 
\qquad
D_0 : f \mapsto \left( y \in \R^n \mapsto \sum_{0 \leq i \leq n} y_i \frac{\partial^{} f}{\partial x_i}(0) \right)\]
where $(x_i)_i$  is the canonical base of $\R^n$, $y_i$ is the $i$-th coordinate of $y$ in the base $(x_i)_i$, and~$a_{\alpha} \in \R$. To include LPDOcc and $D_0$ in a single semiring structure, one would need to consider global differential operators generated by:

\begin{center}
$\mathsf{D} : f \mapsto \left(  (y,v) \mapsto \sum_{\alpha \in \N^n} a_\alpha(v) \frac{\partial^{|\alpha|} f}{\partial x^{\alpha}}(y) \right), \text{ with $a_{\alpha} \in \Ci(\R^n,\R)$.}$
\end{center}
The algebraic structure of such a set would be more complicated, and the composition in particular would not be commutative, and as such, not suitable for the first law of a semi-ring, which is essential since it ensures the symmetry of the contraction and the cocontraction.

\paragraph{Graded promotion with differential operators}\label{ssec:grad_prom}

To introduce a promotion law in $\IDiLL$, we need to define a multiplicative law $\odot$ on $\D$, with~$D_0$ as a unit. We will write it under a digging form:

\begin{center}
\AXC{$\vdash \Gamma, \wn_{D_1} \wn_{D_2} A$}
\RL{$\mathsf{dig}$}
\UIC{$\vdash \Gamma, \wn_{D_1 \odot D_2} A$}
\DisplayProof
\end{center}

This relates to recent work by Kerjean and Lemay \cite{kerjean_lemay23}, inspired by preexisting mathematical work in infinite-dimensional analysis \cite{Gannoun_dualite}. 
They show that in particular quantitative models, one can define the exponential of elements of $\oc A$, such that $e^{D_0} : \Ci(\R^n,\R) \to \Ci(\R^n,\R)$ is the identity. It hints at a possible definition of the multiplicative law as $D_1 \odot D_2 := D_1 \circ e^{D_2}$.

Even if one finds a semi-ring structure on the set of all LPDOcc, the introduction of promotion in the syntax means higher-order functions in denotational models. Indexed exponential connectives are defined so far thanks to the action of LPDOcc on functions with a finite number of variables. To make LPDOcc act on higher order function (\textit{e.g.} elements of $\Ci(\Ci(\R^n,\R),\R)$ and not only $\Ci(\R^n,\R)$) one would need to find a definition of partial differential operators independent from any canonical base, which seems difficult. Moreover, contrary to what happens regarding the differentiation of the composition of functions, no higher-order version of the chain rule exists for the action of LPDOcc on the composition of functions. A possible solution could come from differentiable programming \cite{mazza_diffprog}, in which differentials of first-order functions are propagated through higher-order primitives.

As a trick to bypass some of these issues, we could consider that the $!_D$ modalities are not composable. This is possible in a framework similar to the original $\BLL$ or that of $\IndLL$~\cite{ill}, where indexes have a source and a target.

\subsection{Other perspectives}
There are several directions to explore now that the proof theory of $\DBSLL$ has been established. The obvious missing piece in our work is the \emph{categorical axiomatization} of our model. In a version with promotion, that would consist of a differential version of bounded linear exponentials \cite{corequantitative}. A first study based on differential categories \cite{diffcat} was recently done by Pacaud-Lemay and Vienney \cite{lemay2023graded}. While similarities will certainly exist in categorical models of $\DBSLL$, differences between the dynamics of LPDOcc and those of differentiation at $0$ will certainly require adaptation. In particular, the treatment of the sum will require attention (proofs do not need to be summed here, while differential categories are additive). Finally, beware that our logic does not yet extend to higher-order and that, without a concrete higher model, it might be difficult to design elegant categorical axioms. 

Another line of research would consist of introducing more complex differential operators as indices of exponential connectives. Equations involving LPDOcc are extremely simple to manipulate as they are solved in a single step of computation (by applying a convolution product with their fundamental solution). The vast majority of differential equations are difficult, if not impossible, to solve. One could introduce fixpoint operators within the theory of $\DBSLL$, to try and modelize the resolution of differential equations by fixed point. This could also be combined with the study of particularly stable classes of differential operators, as \emph{D-finite operators}. We would also like to understand the link between our model, where exponentials are graded with differential operators, and another new model of linear logic where morphisms correspond to linear or non-linear differential operators~\cite{jets}.

The need for $\ast$-autonomous structure is not surprising from a mathematical point of view, as reflexive spaces are central in distribution theory. It is, however, unexpected from a logical point of view, as a traditional graded exponential does not need an involutive duality and can be described in the setting of Intuitionistic Linear Logic. We suggest that a categorical exploration of the interactions between differentiation and $\ast$-autonomy might help us understand potential generalisations of the present work to higher orders. In particular, as mentioned several times in this paper, the isomorphism $E \simeq E''$ is frequently overlooked. While the dual of a graded "of course" $\oc_a E$, for $a$ in a monoid or a semi-ring, should be a graded "why not" $\wn_{a^\ast} E'$, nothing a priori enforces $a=a^\ast$. Works by Ouerdiane~\cite{Gannoun_dualite}, in particular, feature higher-order functions bounded by exponential $e^{\theta}$ where $\theta$ is a Young function. These young functions are also indices for interpretations of $\oc$ and $\wn$ on DF and Fréchet spaces, and duality transforms an index $\theta$ into its convex conjugate $\theta^{\ast}$. We gather that higher-order functional analysis has much to offer on the topic of graded exponentials.

\bibliographystyle{alphaurl}
\bibliography{biblio}

\end{document}